\documentclass[10pt,letter]{article}

\usepackage{amsmath} 
\usepackage{amsfonts} 
\usepackage{amssymb}
\usepackage{amsthm}
\usepackage{eucal}
\usepackage{mathrsfs}
\usepackage{enumerate}
\usepackage{graphicx} 

\usepackage{alltt}
\usepackage{setspace}
\usepackage{times}

\usepackage{array}
\newcolumntype{S}{>{\centering\arraybackslash} m{.4\linewidth} }

\usepackage{hyperref}
\usepackage{natbib}

\newcommand{\R}{\mathbb{R}}

\newcommand{\E}{\operatorname{\mathbb{E}}}
\newcommand{\Var}{\operatorname{Var}}
\newcommand{\I}{\mathbb{I}}
\renewcommand{\Pr}{\mathbb{P}}

\newcommand{\balpha}{\boldsymbol{\alpha}}

\newcommand{\bvarphi}{\boldsymbol{\varphi}}
\newcommand{\btau}{\boldsymbol{\tau}}

\newcommand{\bsigma}{\boldsymbol{\sigma}}

\newcommand{\PhiT}[3]{\varPhi_{#1,#3}^{#2}}
\newcommand{\bPhiT}{\boldsymbol{\mathrm{\Phi}}}
\newcommand{\nPhiT}[3]{\hat{\varPhi}_{#1,#3}^{#2}}
\newcommand{\bnPhiT}{\boldsymbol{\hat{\Phi}}}
\newcommand{\bnPsiT}{\boldsymbol{\hat{\Psi}}}
\newcommand{\nPsiT}[3]{\hat{\Psi}_{#1,#3}^{#2}}

\renewcommand{\vec}[1]{\mathbf{#1}}
\newcommand{\mat}[1]{\mathbf{#1}}

\newcommand{\abs}[1]{\left\vert#1\right\vert}

\newtheorem{thm}{Theorem}[section]
\newtheorem{corol}[thm]{Corollary}

\newtheorem{lemma}[thm]{Lemma}
\newtheorem{prop}[thm]{Proposition}
\theoremstyle{definition}

\theoremstyle{remark}
\newtheorem{remark}[thm]{Remark}

\def\ITMProbe{\textsf{{\footnotesize\itshape ITM Probe}}}

\begin{document}
\begin{titlepage}
\begin{center}
{\Large\bf Information flow in interaction networks II: channels, path lengths and potentials}
\end{center}
\vspace{.35cm}

\begin{center}
{\large Aleksandar Stojmirovi\'c\, and Yi-Kuo Yu\footnote{to whom correspondence should be addressed}}
\vspace{0.25cm}
\small

\par \vskip .2in \noindent
National Center for Biotechnology Information\\
National Library of Medicine\\
National Institutes of Health\\
Bethesda, MD 20894\\
United States
\end{center}

\normalsize
\vspace{0.25cm}

\begin{abstract}
In our previous publication, a framework for information flow in interaction networks
based on random walks with damping was formulated with two fundamental modes:
emitting and absorbing. While many other network analysis methods based on
random walks or equivalent notions have been developed before and after our
earlier work, one can show that they can all be mapped to one of the two modes.
In addition to these two
fundamental modes, a major strength of our earlier formalism was its accommodation
of context-specific {\it directed} information flow that yielded plausible and
meaningful biological interpretation of protein functions and pathways.
However, the directed flow from origins to destinations
was induced via a potential function that was heuristic.
Here, with a theoretically sound approach called the \emph{channel mode},
we extend our earlier work for directed information flow. This is achieved by
constructing a potential function facilitating a purely probabilistic interpretation of the channel mode.
For each network node, the channel mode combines the solutions of emitting
and absorbing modes in the same context, producing what we call a \emph{channel tensor}.
The entries of the channel tensor at each node can be interpreted as the amount of flow
passing through that node from an origin to a destination. Similarly to our earlier model,
the channel mode encompasses damping as a free parameter that controls
the locality of information flow. Through examples involving the yeast pheromone
response pathway, we illustrate the versatility and stability of our new framework.
\end{abstract}
\end{titlepage}

\section{Introduction}

Biological pathways in protein interaction networks have been modelled \citep{TWACS06,SY07,SBKEI08} as information flow or equivalently random walks between pathway origins and destinations. Ideally, the nodes visited by the flow should suggest a mechanism for the pathway being investigated. For biological specificity of the results, it is important that the flow is directed and localized, that is, the random walks should follow more direct paths from origins to destinations, as opposed to wandering around the whole network. Otherwise, if pathway origins and destinations are distant, many proteins (particularly large network hubs) unrelated to the pathway's biological function may appear as significant. It is therefore necessary to construct a model that is able to controllably pull the information flow towards the pathway destinations.

In an earlier paper~\citep{SY07}, we developed a mathematical framework that is capable of directing information flow in interaction networks based on random walks. Via information damping/aging, this framework naturally accommodates information loss/leakage that always occurs in all networks. It requires no prior restriction to the sub-network of interest nor it uses additional (and possibly noisy) information. The framework consisted of two modes \emph{absorbing} and \emph{emitting}. Given a set of information \emph{sinks}, the absorbing mode returns for any network node the likelihood of a random walk starting at that node to terminate at sinks. The emitting mode returns for each network node the expected number of visits to that node by a random walk starting at information \emph{sources}. The emitting mode can also be used to model biological pathways: given sources and selected destinations (pseudosinks), we introduced heuristic potential functions that adjust the weights of network links to guide the information flow towards pseudosinks~\citep{SY07}.

Although the introduction of potential to direct information flow is novel, the concepts of diffusion and random walks have been extensively used for analysis of protein interaction networks. \citet{NJAC05} introduced an algorithm that used truncated diffusion from nodes in interactomes to predict protein function. \citet{TWACS06} used simulations of random walks to infer gene regulatory pathways, while \citet{SBKEI08} modelled the interactome as an electrical network to interpret expression quantitative loci (eQTLs). The latter two approaches are conceptually similar due to the correspondence between random walks on (undirected) graphs and electrical networks~\citep{DS84}. \citet{MLZR09} used the electrical network approach to measure network centrality of each node in several interactomes. \citet{VTX09} proposed a spectral measure of closeness between two proteins based on PageRank to discover functionally related proteins. Most efforts in this direction -- for example, the methods proposed by \citet{SBKEI08}, \citet{MLZR09} and \citet{VTX09} -- can be mapped to our absorbing and emitting modes, without potentials (see Section \ref{subsec:oldinterpret} for details).

While our earlier model provides very reasonable results on many examples from yeast protein-protein interaction networks~\citep{SY07}, it also has room for improvement. Absent a theory, the potential functions were empirically chosen and the optimal potentials became example-dependent. That is, different potentials might be needed for different networks, sources and pseudosinks. Consequently, the model values (visits) for each node can not be directly interpreted but only in relation to each other. Furthermore, since each choice of the origins and destinations results in a different network graph, rapid computation at large-scale is hindered.

In this sequel, we present a major extension of our previous framework. By appropriately combining the emitting and absorbing modes, we have devised a new, \emph{channel}, mode that permits directed information flow with probabilistic interpretation. The manuscript is structured as follows. Section \ref{sec:techbkgrnd} presents a succinct review of our previous work and shows how other proposed methods can be mapped to its absorbing or emitting mode. Section \ref{sec:theory} details our extension. Section \ref{sec:application} discusses applications of the channel mode to protein interaction networks using the yeast pheromone response pathway as an example. Discussion and conclusions are in Section \ref{sec:discussion}, with more technical details provided in the Appendices.

\section{Technical Background}\label{sec:techbkgrnd}

\subsection{Preliminaries}

We will closely follow the notation of \citet{SY07}. We represent an interaction network as a weighted directed graph $\Gamma=(V,E,w)$ where $V$ is a finite set of vertices of size $n$, $E\subseteq V\times V$ is a set of edges and $w$ is a non-negative real-valued function on $V\times V$ that is positive on $E$, giving the weight of each edge (the weight of non-existing edge is defined to be $0$). Assuming an ordering of vertices in $V$, we represent a real-valued function on $V$ as a state (column) vector $\vec{\bvarphi}\in\R^n$ and the connectivity of $\Gamma$ by the \emph{weight} matrix $\mat{W}$ where $W_{ij}=w(i,j)$ (the weight of an edge from $i$ to $j$). We do not make distinction between a vertex $v\in V$ and its corresponding state given by a particular ordering of vertices.
Denote by $\mat{P}$ the $n\times n$ matrix such that for all $i,j\in V$,
\begin{equation}\label{eq:pmatrix}
P_{ij} = \frac{\alpha_i W_{ij}}{\sum_k W_{ik}},
\end{equation}
when $\sum_{k\in V} W_{ik} > 0$ and $P_{ij}=0$ otherwise. Here $\alpha_i\in(0,1]$ for all $i$.

When $\alpha_i=1$ for all $i$, the matrix $\mat{P}$ is a transition matrix for a random walk or a Markov chain on $\Gamma$: for any pair of vertices $i$ and $j$, $P_{ij}$ gives the transition probability from vertex $i$ to vertex $j$ in one time step. In the general case, the node-specific damping factors $\alpha_i$ model \emph{dissipation} of information: at each step of the random walk there is some probability that the walk leaves the graph. The value $\alpha_i$ measures the likelihood for the random walk leaving the vertex $i$ to remain in the graph, or equivalently, the likelihood of dissipation at $i$ is $1-\alpha_i$.

For this paper, it will be convenient to express dissipation in terms of the uniform damping coefficient $\mu$, where 
\begin{equation}
\mu = \max_{i} \alpha_i.
\end{equation}
Let $a_i=\alpha_i/\mu$ and define the matrix $\mat{Q}$ by $\mat{P}=\mu\mat{Q}$, that is,
\begin{equation}\label{eq:pmatrix1}
Q_{ij} = \frac{a_i W_{ij}}{\sum_k W_{ik}},
\end{equation}
for $i,j\in V$ by and $0<a_i\leq 1$. We will consider $\mu$ as a free parameter in $(0,1]$ and the matrix $\mat{P}$ as dependent on $\mu$.

\subsection{Emitting and absorbing modes}

We extract the properties of information flow through a given network by examining the paths of discrete random walks. A random walker starts at an originating node, chosen according to the application domain, and traverses the network, visiting a node at each step. Each walk terminates at an explicit  \emph{boundary} vertex or due to dissipation, which is modeled as reaching an implicit (out-of-network) boundary node. 

We distinguish two types of boundary nodes: \emph{sources} and \emph{sinks}. Sources emit information, that is, serve as the origins of random walks. All information entering a source from inside the network is dissipated, so a walker is not allowed to visit the source more than once. Sinks absorb information, serving as destinations of walks; information leaving each sink is completely dissipated. The network graph together with a set of boundary nodes and a vector of damping factors $\balpha$ provides the \emph{context} for the information flow investigated.

The main variable of interest is the (averaged) number of times a vertex is visited by a random walk given the context. Let $D$ denote the set of selected boundary nodes, let $T=V\setminus D$ and let $m=\abs{T}$. Assuming that the first $n-m$ states correspond to vertices in $D$, we write the matrix $\mat{P}$ in the canonical block form:
\begin{equation}\label{eqn:Pcannonical}
\mat{P}=\left[ \begin{array}{cc}\mat{P}_{DD} & \mat{P}_{DT}\\ \mat{P}_{TD} & \mat{P}_{TT}\end{array}\right].
\end{equation}
Here $\mat{P}_{AB}$ denotes a matrix giving probabilities of moving from nodes in $A$ to nodes in $B$ 
where $A,B$ stand for either $D$ or $T$. 
The states (vertices) belonging to the set $T$ are called \emph{transient}.

\subsubsection{Absorbing mode}\label{subsec:absorbing}

Suppose that the boundary set $D$ consists only of sinks. Let $\mat{F}$ denote an $m \times (n-m)$ matrix such that $F_{ij}$ is the total probability that the information originating at $i\in T$ is absorbed at $j\in D$. The matrix $\mat{F}$ is found by solving the discrete Laplace equation
\begin{equation}\label{eqn:sink2}
(\I-\mat{P}_{TT})\mat{F} = \mat{P}_{TD},
\end{equation}
where $\I$ denotes the identity matrix. The matrix $\Delta(\mat{P}_{TT})=\I-\mat{P}_{TT}$ is known as the discrete Laplace operator of the matrix $\mat{P}_{TT}$. If $\I-\mat{P}_{TT}$ is invertible, Equation (\ref{eqn:sink2}) has a unique solution
\begin{equation}\label{eqn:sink4}
\mat{F} = \mat{G}\mat{P}_{TD},
\end{equation}
where $\mat{G}=(\I-\mat{P}_{TT})^{-1}$. 

\subsubsection{Emitting mode}\label{subsec:emitting}

Now consider the dual problem where $D$ is a set of sources. Let $\mat{H}$ denote an $(n-m) \times m$ matrix such that $H_{ij}$ is the total expected number of times the transient vertex $j$ is visited by a random walk emitted from source $i$ (for all times). Again, $\mat{H}$ is found by solving the discrete Laplace equation
\begin{equation}\label{eqn:source2}
\mat{H}(\I-\mat{P}_{TT}) = \mat{P}_{DT}.
\end{equation}
which, if $\I-\mat{P}_{TT}$ is invertible, has a unique solution
\begin{equation}\label{eqn:source4}
\mat{H} = \mat{P}_{DT}\mat{G}.
\end{equation}

It is easy to show \citep{SY07} that the matrix $\mat{G}=(\I-\mat{P}_{TT})^{-1}$, also known as the Green's function or the fundamental matrix of an absorbing Markov chain \citep{KS76}, exists if every node can be connected to a boundary node or if $\alpha_i<1$ for all $i$. The entry $G_{ij}$ represents the mean number of times the random walk reaches vertex $j\in T$ having started in state $i\in T$ \citep{KS76}. For any transient state $i$, the value
\begin{equation}
T_i = \sum_{j\in T} G_{ij}
\end{equation}
gives the average length of a path traversed by a random walker starting at $i$ before terminating~\citep{KS76}. In this case, the walker is allowed to revisit $i$ after leaving $i$. In the Markov chain theory, $T_i$ is also known as the average absorption time from $i$. For the emitting mode, where the walker starts at $s\in S$ and cannot revisit it, it can be shown that the average path length is
\begin{equation}
T_s = 1+\sum_{j\in T} H_{sj}
\end{equation}

\subsection{Interpretations}\label{subsec:oldinterpret}

If we assume that a random walk deposits a fixed amount of information content each time it visits a node, we can interpret $H_{ij}$ is the overall amount of information content originating from the source $s$ deposited at the transient vertex $j$. Furthermore, we can interpret $F_{ij}$ as the sum of probabilities (weights) of the paths originating at the vertex $i\in T$ and terminating at the vertex $j\in D$ while avoiding all other boundary nodes in the set $D$, and $H_{ij}$ as the sum of probabilities (weights) of the paths originating at the vertex $i\in D$ and terminating at the vertex $j\in T$, also avoiding all other nodes in the set $D$. Each such path has a finite but unbounded length. However, unlike $F_{ij}$, $H_{ij}$ does not represent a probability because the events of the information being located at $j$ at the times $t$ and $t'$ are not mutually exclusive (a random walk can be at $j$ at time $t$ and revisit it at time $t'$). For $F_{ij}$, the absorbing events at different times are mutually exclusive.

The entry $H_{ij}$ can alternatively be interpreted as equilibrium information content at $j$ for information flow originating from $i$. In this case we imagine the flow entering the network at node $i$ and leaving the network at $i$ and any other node due to dissipation. The amount of inflow at $i$ is set to $1$ and $H_{ij}$ denotes the steady state content at $j$. Hence, the \emph{equilibrium flow rate} through an edge $(i,j)$ by the flow entering at $s\in D$, denoted $\psi_s(i,j)$, is 
\begin{equation}
\psi_s(i,j) = H_{si}P_{ij}.
\end{equation}

\subsubsection{Electrical networks and heat conduction}

A weighted undirected graph $\Gamma=(V,E,w)$ can  be considered as an electrical network with each edge weight $(i,j)$ being associated with resistance $R_{ij}=1/W_{ij}$. \citet{DS84} have shown that voltages and currents through nodes and edges can be interpreted in terms of random walks with transition matrix $\mat{P}$ (where $\alpha_i=1$ for all $i\in V$) and absorbing boundary. Let $\vec{f}$ denote the voltage vector over all nodes and suppose that a unit voltage is applied between two nodes $a$ and $b$, so that $f_a=1$ and $f_b=0$. Then, the solution for $\vec{f}$ over $T=v\setminus\{a,b\}$ according to Kirchhoff's laws is equivalent to the $a$-th column of the absorbing mode matrix $\mat{F}$, that is, $f_i = F_{ia}$. The current flowing through an edge $(i,j)$, which we denote $I_{ij}$, is then given by 
\begin{equation}
I_{ij} = \frac{f_i - f_j}{R_{ij}} = (F_{ia}-F_{ja})W_{ij}.
\end{equation}
Therefore, modeling protein interaction networks as resistor networks is equivalent to applying our absorbing mode without dissipation. 

However, electrical network paradigm is only applicable to interaction networks where all links can be modeled as undirected edges. This is the case in \citep{MLZR09}, where the authors only take physical interactions between proteins as links in their networks. On the other hand, the network constructed by 
\citet{SBKEI08} contained, in addition to physical interactions, the transcription factor-to-gene interactions. These interactions were modeled as directed edges and \citet{SBKEI08} applied a heuristic approach to model the current flowing through them. In contrast, our absorbing mode can be directly applied to directed networks, although the columns of the matrix $\mat{F}$ cannot be interpreted as voltages (Figure~\ref{fig:example1}). We will show in \ref{subsubsec:operator} that, even when edges are directed, $\mat{F}$ gives rise to potentials.

\begin{figure}[hp]
\begin{center}
\includegraphics{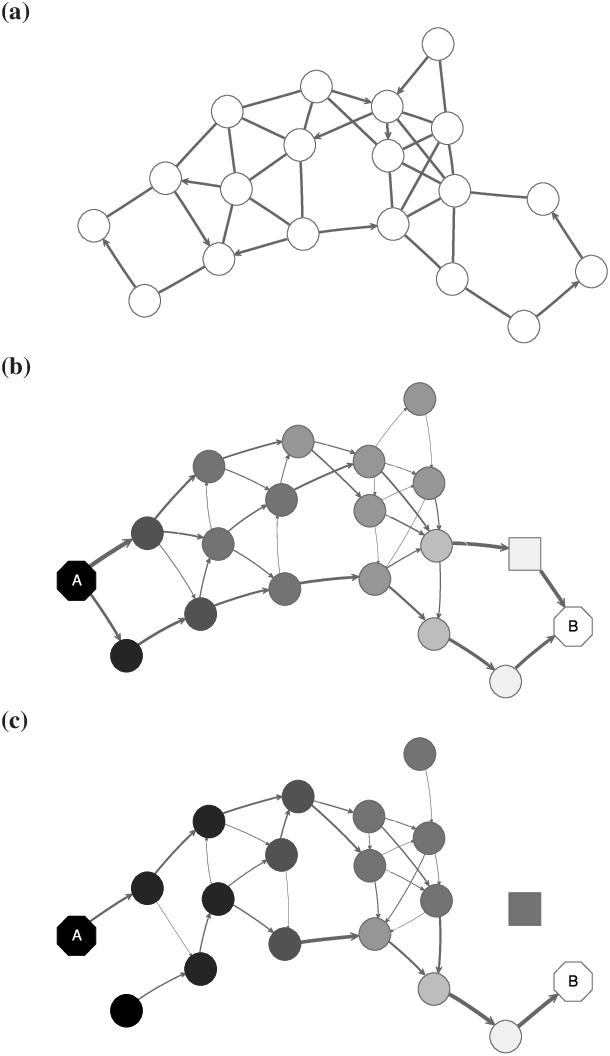}
\caption[Absorbing mode formalism can be extended beyond resistor networks.]{Absorbing mode formalism can be extended beyond resistor networks. Consider, for example, the directed graph shown in \textbf{(a)}, where all edges, directed and undirected have weight 1. This graph can be modeled as a resistor network by treating all edges as undirected:  \textbf{(b)}. Applying a unit voltage at node A and grounding at node B leads to the current flowing from A to B. The voltage at each node is indicated by shading (dark means high voltage) while the current at each edge is indicated by the thickness and the direction of the arrow corresponding to that edge. The equivalents of voltage and current can be obtained for the original graph using the absorbing mode with the same boundary: \textbf{(c)}. Note the qualitative difference between the results in \textbf{(b)} and \textbf{(c)}: the node shaped as square conducts significant current in \textbf{(b)} but is totally isolated in \textbf{(c)}.  }\label{fig:example1}
\end{center}
\end{figure}

\citet{ZBY07} applied the same formalism without damping to social networks as a recommendation model. They consider a graph $\Gamma$ corresponding to a social network, where items under consideration are mapped to nodes, as a heat conduction medium and interpret $f$ as temperature. For each recomendee, by setting his/her favorite items to `high-temperature' and disliked items to `low-temperature' and solving for $f$ over the remaining nodes, they obtain the heat distribution over the entire $\Gamma$. The values of $f$ can be used to recommend potential interesting items (other high temperature nodes) to individuals.

\subsubsection{Topic-sensitive PageRank} Topic-sensitive PageRank was introduced by \citet{Haveliwala03} as a context sensitive algorithm for web search and has been recently applied to protein interaction networks by \citet{VTX09}. The PageRank vector $\vec{p}$ is defined as the unique solution of the equation
\begin{equation}
\vec{p} = \beta\vec{s} + (1-\beta)\vec{p}\mat{M},
\end{equation}
where $\mat{M}$ is the transition matrix for a graph (i.e. $\sum_{j\in V} M_{ij}=1$), $0<\beta < 1$ and $\vec{s}$ is a probability vector ($\sum_j s_j = 1$). The vector $\vec{p}$ is interpreted as the steady state for the random walk governed by $\mat{M}$, which at each step has probability $\beta$ of restarting at a different node. The probability of restarting at the node $j$ is $s_j$. Clearly, $\vec{p}$ can be written as
\begin{equation}\label{eqn:pagerank2}
\vec{p} = \beta\vec{s}(\I-(1-\beta)\mat{M})^{-1}.
\end{equation}
PageRank can be considered a special case of the emitting mode in the following way. Add an additional vertex $v$ to the graph with no incoming edges and with the weight of each outgoing edge $v\to i$ proportional to $s_i$. Construct a matrix $\mat{P}$ using $\alpha_i = 1-\beta$ for all $i$ in the original graph and $\alpha_v=\beta$. Let $D=\{v\}$ be the boundary set. Clearly, $(1-\beta)\mat{M}=\mat{P}_{TT}$ and $\beta\vec{s}=\mat{P}_{DT}$, and hence Equation~\eqref{eqn:pagerank2} reduces to Equation~\eqref{eqn:source2}.

\subsubsection{Other methods based on random walks}

Beyond the analysis of protein interaction networks, approaches based on diffusion and random walks have received attention for a number of applications. We will only mention here a few examples from machine learning to illustrate the point.

A \emph{kernel} on a space $X$ is a symmetric positive (semi)definite map $\kappa:X\times X\to \R$, which can be used to measure similarity between two points in $X$. A kernel can naturally be treated as an inner product on some feature space. Among other approaches, kernels are the foundation of Support Vector Machines (SVMs), machine learning methods widely used for classification and pattern recognition of data~\citep{SS02,STV04}. 

A variety of kernels were proposed to compare nodes in undirected graphs~\citep{FYPS06}, mostly derived from discrete Laplacians. Recall that we called the matrix $\Delta(\mat{P}_{TT})=\I-\mat{P}_{TT}$ the discrete Laplace operator of the matrix $\mat{P}_{TT}$. One can similarly define the matrices $\Delta(\mat{W})=\I-W$ and $\Delta(\mat{P})=\I-\mat{P}$, where $W$ is the adjacency matrix and $\mat{P}$ is the transition matrix for a weighted undirected graph $\Gamma$. Both $\Delta(\mat{W})$ and $\Delta(\mat{P})$ were sometimes called the graph Laplacians for $\Gamma$.

Generally, the matrix $\Delta(\mat{W})$ need not be invertible (in particular, $\Delta(\mat{P})$ is not invertible -- see \citep{ZBY07}). \citet{FPRS07} proposed using the Moore-Penrose pseudoinverse, which generalizes a matrix inverse to matrices of less than full rank, of $\Delta(\mat{W})$ as a kernel, with applications to collaborative recommendation. The approach and the application domain of \citet{FPRS07} are similar to that of \citet{ZBY07}.

The von Neumann diffusion kernel~\citep{SS02}, proposed by \citet{Katz53} has the form
\begin{equation}
\kappa = \sum_{n=1}^\infty \beta^n [\mat{W}^n] = (\I - \beta\mat{W})^{-1} - \I,
\end{equation}
where $\beta$ is a damping factor chosen so that $(\I - \beta\mat{W})^{-1}$ exists. This approach is roughly similar to ours where we compute $\mat{G} = (\I - \mu\mat{Q}_{TT})^{-1}$ in that both $\kappa_{ij}$ and $G_{ij}$ include the sums of the weights for all paths from $i$ to $j$. The main difference between the two approaches is that the weight of each path of length $n$ included in $\kappa$ is the product of weights of each link followed, while in our case it is the product of probabilities and therefore has a probabilistic interpretation.

Exponential diffusion kernels, introduced by \citet{KL02}, are defined by
\begin{equation}
\kappa = \sum_{n=0}^\infty \frac{\beta^k(-\Delta(\mat{W}))^k}{k!} = \exp(-\beta\Delta(\mat{W})),
\end{equation}
with a real parameter $\beta$. Diffusion kernels can be interpreted to model continuous diffusion through graph, with infinitesimal time steps in contrast to discrete-time diffusion implied by von Neumann diffusion kernel and other similar random-walk based methods. Note that, since every kernel is required to be symmetric, the above formalizations do not extend directly to directed graphs.

\section{Theory}\label{sec:theory}

Assume $V=S\sqcup T\sqcup K$, where the set $S$ denotes the sources, $K$ denotes the sinks and $T$ the transient nodes and write the matrix $\mat{P}$ in the block form as
\begin{equation}
\mat{P}=\left[ \begin{array}{ccc}\mat{P}_{SS} & \mat{P}_{ST} & \mat{P}_{SK}\\ \mat{P}_{TS} & \mat{P}_{TT} & \mat{P}_{TK}\\ \mat{P}_{KS} & \mat{P}_{KT} & \mat{P}_{KK} \end{array}\right].
\end{equation}
Let us modify (add context to) the underlying graph $\Gamma$ so that the random walk can only leave the sources and only enter the sinks. Furthermore, no communication is allowed among sources or among sinks without going through transient nodes. The modified transition matrix, denoted $\mat{\tilde{P}}$ has the form
\begin{equation}
\mat{\tilde{P}}=\left[ \begin{array}{ccc}\mat{0} & \mat{P}_{ST} & \mat{P}_{SK}\\ \mat{0} & \mat{P}_{TT} & \mat{P}_{TK}\\ \mat{0} & \mat{0} & \mat{0} \end{array}\right].
\end{equation}

Treating the vertices in $S$ and $T$ as transient for the absorbing mode in \ref{subsec:absorbing}, we first derive the matrix $\mat{F}$ (of size $\abs{S\cup T}\times\abs{K}$):
\begin{align*}
\mat{F} &= \left(\I- \left[ \begin{array}{cc}\mat{0} & \mat{P}_{ST}\\ \mat{0} & \mat{P}_{TT}\end{array}\right]  \right)^{-1} \left[ \begin{array}{c} \mat{P}_{SK}\\ \mat{P}_{TK}\end{array}\right] = \left[ \begin{array}{cc}\I & \mat{P}_{ST}\mat{G}\\ \mat{0} & \mat{G}\end{array}\right] \left[ \begin{array}{c} \mat{P}_{SK}\\ \mat{P}_{TK}\end{array}\right]  \\
& = \left[ \begin{array}{cc} \mat{P}_{SK} + \mat{P}_{ST}\mat{G}\mat{P}_{TK} & \mat{G}\mat{P}_{TK}\end{array}\right]^T,
\end{align*}
where, as before, $\mat{G}=(\I-\mat{P}_{TT})^{-1}$. Similarly, treating the vertices in $T$ and $K$ as transient for the emitting mode in \ref{subsec:emitting}, we derive the matrix $\mat{H}$ (of size $\abs{S}\times\abs{T\cup K}$):
\begin{align*}
\mat{H} &= \left[ \begin{array}{cc} \mat{P}_{ST} & \mat{P}_{SK}\end{array}\right] \left(\I- \left[ \begin{array}{cc}\mat{P}_{TT} & \mat{P}_{TK}\\ \mat{0} & \mat{0} \end{array}\right]  \right)^{-1} 
= \left[ \begin{array}{cc} \mat{P}_{ST} & \mat{P}_{SK}\end{array}\right] \left[ \begin{array}{cc}\mat{G} & \mat{G}\mat{P}_{TK}\\ \mat{0} & \mat{\I} \end{array}\right] \\
&= \left[ \begin{array}{cc} \mat{P}_{ST}\mat{G} & \mat{P}_{ST}\mat{G}\mat{P}_{TK}+\mat{P}_{SK} \end{array}\right].
\end{align*}

The entries of $\mat{F}$ and $\mat{H}$ are, as before, interpreted as probabilities of absorption at sinks and average numbers of visits of walks emitted from sources, respectively. Note that the same Green's function, $\mat{G}=(\I-\mat{P}_{TT})^{-1}$,  needs to be computed for both solutions. Also note that the `$S$' rows of $\mat{F}$ form the transpose of the `$K$' columns of $\mat{H}$, that is, for all $s\in S$ and $k\in K$, $F_{sk}=H_{sk}$.

The matrices $\mat{F}$ and $\mat{H}$ can be extended over the whole graph into the matrices $\mat{\bar{F}}$ and $\mat{\bar{H}}$, of sizes $n\times\abs{K}$ and $\abs{S}\times n$, respectively by setting $\bar{F}_{kk'}=\delta_{kk'}$ for $k,k'\in K$ and $\bar{H}_{ss'}=\delta_{ss'}$ for $s,s'\in S$. This is equivalent to setting the $K$ portion of $\mat{\bar{F}}$ and $S$ portion of $\mat{\bar{H}}$ to appropriately sized identity matrices:
\begin{align}
\mat{\bar{F}}& = \left[ \begin{array}{ccc} \mat{P}_{SK} + \mat{P}_{ST}\mat{G}\mat{P}_{TK}, & \mat{G}\mat{P}_{TK}, & \I \end{array}\right]^{T}\label{eq:barF}\\
\mat{\bar{H}}&= \left[ \begin{array}{ccc}\I, & \mat{P}_{ST}\mat{G}, & \mat{P}_{ST}\mat{G}\mat{P}_{TK}+\mat{P}_{SK} \end{array}\right]\label{eq:barH}
\end{align}
The matrices $\mat{\bar{F}}$ and $\mat{\bar{H}}$ contain explicit boundary conditions with interpretations straightforwardly extended from $\mat{{F}}$ and $\mat{{H}}$. Specifically, $\bar{F}_{kk'}=\delta_{kk'}$ means that a random walk originating from a sink cannot move anywhere else, while $\bar{H}_{ss'}=\delta_{ss'}$ implies that a random walk starting at a source will visit it exactly once and cannot return to it nor to any other source.

\begin{remark}
We explicitly assumed that a boundary node can either be a source or a sink. Sometimes, it is desirable to examine flows that both start and terminate at the same point. This case can be reduced to our assumption by introducing for each source that is also a sink an additional node with all the edges of the original node. The new enlarged graph will contain two `logical' nodes for each `physical' source/sink node that plays a dual role and hence it will be possible to have disjoint sets of sources and sinks on the boundary.
\end{remark}

\subsection{Channel tensor}\label{subsubsec:channeltens}

Define the \emph{channel tensor} $\bPhiT\in V\otimes K\otimes S^*$ by
\begin{equation} \label{eq:channeltens}
\PhiT{i}{s}{k} = \bar{H}_{si}\bar{F}_{ik}.
\end{equation}
The entry $\PhiT{i}{s}{k}$ gives the expected number of times a random walk emerging from the source $s$ and terminating at the sink $k$ visits the vertex $i$ (Lemma~\ref{lem:channelexpct}). In particular, for all for all $s\in S$ and $k\in K$, 
\begin{equation}\label{eq:phieq1}
\PhiT{s}{s}{k} = \PhiT{k}{s}{k} = F_{sk} = P_{sk} + [\mat{P}_{ST}\mat{G}\mat{P}_{TK}]_{sk}.
\end{equation}

Hence, the entries of $\bPhiT$ can be interpreted similarly to the entries of $\mat{\bar{H}}$: as expected numbers of visits to nodes in network by random walkers starting at a source node. While $\bar{H}_{si}$ gives the total number of visits to $i$ by a random walker starting at $s$, $\PhiT{i}{s}{k}$ measures only those walkers that ultimately reach the sink $k$. All other walkers, which either terminate due to dissipation before reaching $k$, reach other sinks or reach any of the sources, are not considered. Alternatively, $\PhiT{i}{s}{k}$ measures the amount of equilibrium flow through the node $i$ by a stream of particles entering through $s$ and leaving from $k$. The corresponding equilibrium flow through an edge $(i,j)$, denoted $\psi_{s,k}(i,j)$ is given by $\psi_{s,k}(i,j)=\PhiT{i}{s}{k}P_{ij}$. 

Suppose $s$ and $k$ are connected through a directed path (equivalently $F_{sk} > 0$) and let $T_{sk}$ denote the expected length of the path traversed by a walker starting at $s$ and terminating at $k$. Then, it can be shown (Lemma~\ref{lem:meanpathlen1}) that, 
\begin{equation}
T_{sk} = 1 + \sum_{i\in T}\frac{\PhiT{i}{s}{k}}{F_{sk}} = \frac{\mu}{F_{sk}}\frac{\partial F_{sk}}{\partial\mu}.
\end{equation}
Other moments and cumulants of the distribution of lengths of paths traversed by walkers starting at $s$ and terminating at $k$ can similarly be expressed in terms of the Green's function $\mat{G}$ or the derivatives of $F_{sk}$ with respect to $\mu$ (see Appendix~\ref{apndx:A}).
%

\subsection{Normalized channel tensor}\label{subsubsec:normchanneltens}

For brevity we will use a convention that when a set symbol replaces an ordinary index, it means to sum over that entity index of the set in question. For example, for any $i\in S\cup T$, $F_{iK}\equiv \sum_{k\in K} F_{ik}$ and for all $s\in S$, $i\in V$, $\PhiT{i}{s}{K}\equiv\sum_{k\in K}\PhiT{i}{s}{k}$.

For $s\in S$, $F_{sK}$ gives the probability (or expectation) of a random walk emerging from the source $s$ reaching any of the sinks in $K$. Assuming $F_{sK}>0$ for all $s\in S$, define the \emph{normalized channel tensor}, $\bnPhiT\in V\otimes K\otimes S^*$ by 
\begin{equation}\label{eq:phieq2}
\nPhiT{i}{s}{k} = \frac{\PhiT{i}{s}{k}}{F_{sK}}.
\end{equation}
The normalized channel tensor $\nPhiT{i}{s}{k}$ gives the expectation of the {\it normalized} number 
of visits to $i$ by a random walker from $s$ to $k$. Even though $\PhiT{i}{s}{k}$ in 
(\ref{eq:channeltens}) does not 
consider any of the random walk paths that return to sources or terminate due to dissipation at transient nodes,
 the number of visits to any specific node it records is reduced as the dissipation strength increases.
The normalization by $F_{sK}$ in (\ref{eq:phieq2}) takes out the global effect of damping and makes it possible to
 compare the channel tensors obtained at different dissipation strengths.

\subsection{Interpretations}

Generally, the entries of $\bPhiT$ and $\bnPhiT$ can be interpreted in the same way as the entries of $\mat{H}$ from the emitting mode. For practical applications, it is sometimes desirable to reduce the amount of available information to a single vector over $V$, which can be tabulated and graphically visualized using color maps. 

For a source $s\in S$, the \emph{source specific content} of a node $i\in V$ is $\nPhiT{i}{s}{K}$, the (normalized) total number of visits to $i$ by a random walker starting from $s$ and terminating at any of the sinks in $K$. Equations (\ref{eq:phieq1}-\ref{eq:phieq2}) imply that for all $s\in S$,
\begin{equation}\label{eq:nphitotmass}
\nPhiT{s}{s}{K} = \sum_{k\in K} \nPhiT{k}{s}{k} = 1, 
\end{equation}
that is, the entire flow starting at $s$ and reaching one of the sinks is normalized to unity. The \emph{total content} vector of $\bnPhiT$, denoted by $\hat{\btau}$, sums all (normalized) visits for each node:
\begin{equation}
\hat{\tau}_i = \nPhiT{i}{S}{K}.
\end{equation}
 The concept of \emph{destructive interference} measures the overlap between visits from different sources for each node. We define the interference vector $\hat{\bsigma}$ over $V$ by
\begin{equation}
\hat{\sigma}_i = \abs{S}\min_{s\in S} \nPhiT{i}{s}{K}.
\end{equation} Hence, $\hat{\sigma_i}$ gives the (normalized) total number of times the random walks from all sources co-occur at each node (scaled by the number of sources). The above formulas assume that each source emits the same amount of information. If needed, $\nPhiT{i}{s}{K}$ can be weighted by \emph{source strength} before evaluating total content or interference.

With damping factors less than unity, the random walks from sources to sinks effectively visit a small portion of the entire network. Information Transduction Module or ITM is a notion that we coined to describe the set of nodes most influenced by the flow. The nodes are ranked using their values for the total content or interference and the most significant nodes are selected. The number of selected nodes depends on the application-specific considerations but we found that the \emph{participation ratio} $\pi$ \citep{SY07} of the total content vector $\hat{\btau}$ gives a good
estimate of the number of nodes whose relative amount of content is significant. It is given by the formula
\begin{equation}
\pi(\hat{\btau}) = \frac{\left(\sum_{i\in V} \hat{\tau}_i \right)^2}{\sum_{j\in V} \hat{\tau}_j^2}.
\end{equation}

For undirected graphs, with a context consisting of a single source and a single sink, the values of $\bnPhiT$ are invariant under interchange of sources and sinks (see Appendix~\ref{apndx:reverse}). In general, however, reversing sources and sinks gives a different result, both due to asymmetry of the weight matrix in directed graphs and because sources and sinks have different roles if more than one of each are present: random walkers originating from different sources can simultaneously visit a transient node while a walk can terminate only at a single sink. Thus, the sinks split the total  information flow, that is, compete for it, while sources interfere, either constructively or destructively. 

\subsection{Path lengths}

Damping influences the normalized channel tensor differently from the non-normalized one or the absorbing and emitting solutions. For the non-normalized versions, damping factors control the amount of information reaching the boundary and any intermediate points. In the normalized case, all ``normalized'' information emitted from the sources reaches sinks (Equation \eqref{eq:nphitotmass}) and damping controls a random walker's average path  length, which is always bounded below by the length of the shortest path. Provided each source is connected to at least one sink through a directed path, we have (Corollary~\ref{cor:meanvar2})
\begin{equation}
T_{sK} = 1 + \sum_{i\in T}\nPhiT{i}{s}{K} = \frac{\mu}{F_{sK}}\frac{\partial F_{sK}}{\partial\mu}.
\end{equation}
Small values of $\mu$ strongly favor the nodes on the shortest paths, while large values allow random walks to take longer excursions and hence favor the vertices with many connections. As $\mu\downarrow 0$, only the nodes at the shortest path receive any flow and $T_{sK}\to \rho(s,K)$, the smallest distance between $s$ and any sinks in $K$. Appendix~\ref{apndx:A} contains a more detailed analysis of the role of damping with full statements and proofs.

Note that the $\mu$ dependence of $T_{sK}$ allows one to determine the appropriate damping factor for a specified average path length. From the results in Appendix~\ref{apndx:A}, it follows that $T_{sK}$ is a smooth function of $\mu$, which is strictly increasing on $[0,1]$ ($\frac{\partial T_{sK}}{\partial\mu}$ is positive). Therefore, the equation $T_{sK}(\mu)=x$ has a unique simple root for $\rho(s,K)\leq x\leq T_{sK}(1)$ and any root-finding method can be used to find $\mu$ from $T_{sK}$. When a context contains multiple sources, a desired weighted average of $T_{sK}$ over all $s\in S$ can be set to obtain a global uniform damping factor $\mu$.

\subsection{Potentials and normalized evolution operators}\label{subsubsec:operator}

In \citep{SY07}, we used a concept of a \emph{potential} to redirect the flow towards desired destinations in the emitting mode. To each node $j\in V$, we associated the value of the total potential $\Theta(j)$ such that
\begin{equation}
\Theta(j) = \sum_{k\in R}\theta_k(\rho(j,k)),
\end{equation}
where $R\subset T$ is the set of potential centers, $\rho(j,k)$ is the length of the shortest path from $j$ to $k$, and $\theta_k$ is an increasing function with a minimum at $k$. The exponential of the total potential was then used to re-weight the edges incoming to $j$ and form a new matrix $\hat{\mat{W}}$:
\begin{equation}\label{eq:potn2}
\hat{W}_{ij} = W_{ij}\exp(-\Theta(j)).
\end{equation}
The matrix $\hat{W}$ was then normalized to construct the transition matrix to be used (after applying damping) for the emitting mode. It is possible to express the application of the potential $\Theta$ as a direct transformation of the transition matrix $\mat{P}$ (without dissipation included). Let $f_j\equiv \exp(-\Theta(j))$ and let $\hat{\mat{P}}$ denote the new transition matrix derived from $\hat{\mat{W}}$. Then, $\hat{\mat{P}}$ can be written as
\begin{equation}
\hat{P}_{ij} = \frac{\hat{W}_{ij}}{\sum_{k\in V} \hat{W}_{ik}} = c_i\frac{P_{ij}f_j}{f_i},
\end{equation}
where
\begin{equation}
c_i = \frac{f_i\sum_{k\in V} W_{ik}}{\sum_{k\in V} W_{ik}f_k}.
\end{equation}
If $c_i=1$ for all $i$, we can express $\hat{\mat{P}}$ as a similarity transformation of $\mat{P}$, where $\hat{\mat{P}} = \mat{\Lambda}^{-1}\mat{P}\mat{\Lambda}$, where $\Lambda_{ij} = \delta_{ij}f_i$. In general, this is not the case with the heuristic potentials proposed in \citep{SY07}. However, we will now show (with proofs in Appendix~\ref{app:normop}) that there exist a potential derived from the matrix $\mat{F}$ that transforms the context specific matrix $\mat{\tilde{P}}$ into a stochastic transition matrix over source and transient nodes. The solution of the emitting mode using the new matrix recovers the normalized channel tensor $\bnPhiT$ and allows for  additional generalizations.

Let $V_K=\{i\in V: \bar{F}_{iK} > 0\}$ be the set of all nodes in $V$ that are connected with any sink in $K$ by a directed path and denote by $S_K$ and $T_K$ the sets $S\cap V_K$ and $T\cap V_K$, respectively. Suppose $0\leq \mu \leq 1$. For $i \in S_K\cup T_K$, let 
\begin{equation}\label{eq:genpotential}
f_i = \sum_{k\in K} \bar{F}_{ik} f_k, 
\end{equation} 
where $f_k > 0$ are arbitrary for $k\in K$.
For $i,j\in V_K$, define
\begin{equation}\label{eq:matrixN}
N_{ij} = \frac{\tilde{P}_{ij}f_j}{f_i}.
\end{equation}  
Since all transient nodes are assumed to be connected to a sink, the matrix $\mat{N}$ is well defined for $0 < \mu \leq 1$.  It can be shown using parts of Appendix~\ref{subsec:damping0} that it is also well defined in the limit as $\mu\downarrow 0$. Clearly, $N_{kj} = 0$ for all $k\in K$ and $j\in V_K$. Over $S_K\cup T_K$, the matrix $\mat{N}$ is stochastic (Proposition~\ref{prop:normop2a}), that is $\sum_{j\in V_K} N_{ij} = 1$. Hence, $\mat{N}$ is an evolution operator for flow entering at sources and terminating exclusively at a point in $K$. The dependence on $\mu$ is built in the transition probabilities $N_{ij}$. Furthermore, Equation~\eqref{eq:genpotential} is the only way to construct a function over $V_K$ so that \eqref{eq:matrixN} gives a stochastic transition matrix (Proposition~\ref{prop:normop2a}).

\begin{figure}[htbp]
\begin{center}
\includegraphics{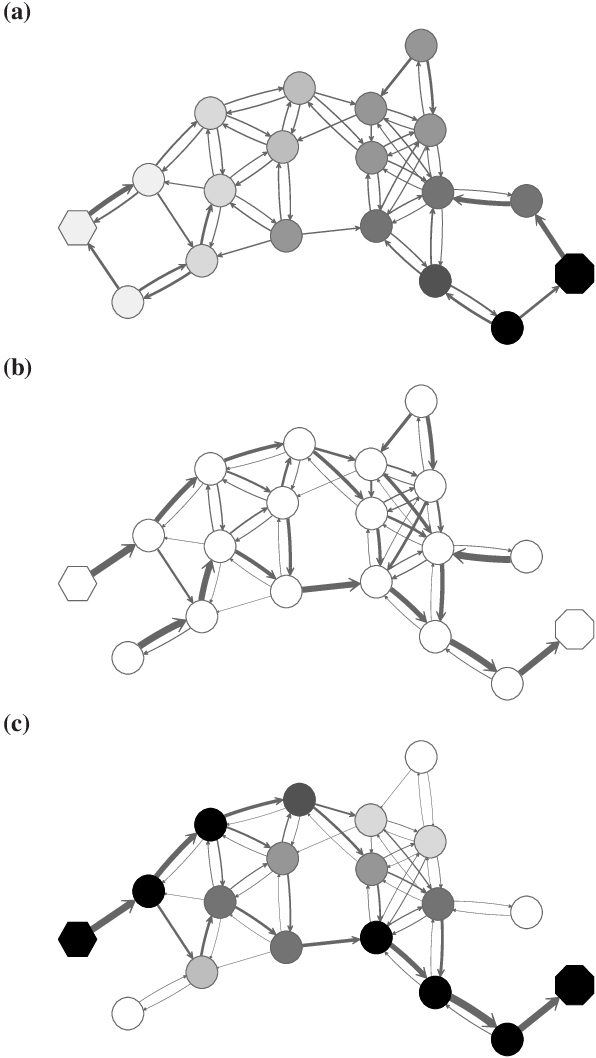}
\caption[Transformation of the evolution operator using potentials]{Transformation of the evolution operator using potentials. Part \textbf{(a)} shows the directed graph from Figure~\ref{fig:example1} with transition probabilities indicated by edge arrows. Nodes are shaded according to the potential associated with the sink (octagon). Part \textbf{(b)} displays the normalized transition operator $\mat{N}$ resulting from the application of the sink potential to the context specific transition matrix (the single source is indicated as hexagon). Part \textbf{(c)} shows the values of the normalized channel tensor as shades and the directional flow through each edge as arrows. Comparison between \textbf{(b)} and \textbf{(c)} shows that edges with large transition probabilities need not carry significant flows.}\label{fig:example2}
\end{center}
\end{figure}

Denote by $\mat{G}(\mat{N})$, $\mat{\bar{F}}(\mat{N})$, $\mat{\bar{H}}(\mat{N})$, $\bPhiT(\mat{N})$ the quantities corresponding to $\mat{G}$, $\mat{F}$, $\mat{H}$ and $\bPhiT$ respectively, when the matrix $\tilde{\mat{P}}$ is replaced by the transition matrix $\mat{N}$. Since transformation \eqref{eq:matrixN} is a similarity transformation from $\tilde{\mat{P}}$ to $\mat{N}$, it is easy to establish
\begin{prop}\label{prop:normop}
The following identities hold:
\begin{enumerate}[(i)]
\item For all $i,j\in T_K$,\quad $\displaystyle [\mat{G}(\mat{N})]_{ij} = G_{ij}f_j/f_i$,
\item For all $i\in V_K$ and $k\in K$,\quad $\displaystyle [\mat{\bar{F}}(\mat{N})]_{ik} = \bar{F}_{ik}f_k/f_i$,
\item For all $s\in S_K$ and $i\in V_K$,\quad $\displaystyle [\mat{\bar{H}}(\mat{N})]_{si} = \bar{H}_{si}f_i/f_s$,
\item For all $s\in S_K$, $i\in V_K$ and $k\in K$,\quad $\displaystyle [\bPhiT(\mat{N})]^{s}_{i,k} = \PhiT{i}{s}{k}f_k/f_s$.
\end{enumerate}
\end{prop}

The special case where $f_k$'s are equal for all $k\in K$ results in $[\mat{\bar{H}}(\mat{N})]_{si}=\nPhiT{i}{s}{K}$ and $[\bPhiT(\mat{N})]^{s}_{i,k}=\nPhiT{i}{s}{k}$. Hence, $\mat{N}$ in this case can be considered a `natural' transition operator for random walks or Markov chains that start at sources $S$ and terminate at a point in $K$. The time evolution of such processes can be followed by raising $\mat{N}$ to appropriate powers. As demonstrated in the previous sections, the parameter $\mu$, which is implicit in $\mat{N}$, controls the how fast the random walkers move towards their destinations. Figure~\ref{fig:example2} shows a graphical example of the transformation of the operator $\tilde{\mat{P}}$ into $\mat{N}$, which directs the flow towards the sink.

In general, each value $f_k$ represents the \emph{sink strength} of the sink $k\in K$. Equal sink strengths imply no prior preference for any sink while in the case of unequal sink strengths the flow from sources towards sinks is preferentially pulled towards sinks with larger strength. It is also possible to exclude some sinks from consideration by setting their strength to $0$. Since the scaling of $f_k$'s does not affect the transition matrix, they can be considered as probabilities over $K$ and, in the Bayesian framework, as priors. Indeed, the equation
\begin{equation}\label{eq:bayes1}
[\mat{\bar{F}}(\mat{N})]_{ik} = \frac{\bar{F}_{ik}f_k}{\sum_{k'\in K}\bar{F}_{ik'}f_{k'} }
\end{equation}
can be easily recognized as Bayes' formula for posterior likelihood. Here $\bar{F}_{ik}$ can be interpreted as the likelihood of a random walk from $i$ being absorbed at sink $k$, given that $k$ is absorbing; $f_k$ is the prior probability that $k$ is absorbing; while $[\mat{\bar{F}}(\mat{N})]_{ik}$ is the likelihood that a walker starting at $i$ is absorbed at $k$, given that it is absorbed at any of the `active' sinks (i.e. sinks with $f_k > 0$). This suggests a use of the absorbing and channel modes as Bayesian inference frameworks for forming and testing hypotheses. For example, sinks can be associated with mutually exclusive hypotheses. The likelihood of each source being associated with a hypothesis can then be evaluated using \eqref{eq:bayes1}.

The matrix $\mat{N}$ can also be expressed in terms of potentials. Suppose $f_k > 0$ for each $k\in K$ and set the potential of each node $i\in V_K$ by
\begin{equation}
\Theta(i)\equiv -\log \sum_{k\in K} F_{ik} f_k.
\end{equation}
Then, $\mat{N}$ can be written as 
\begin{equation}
N_{ij} = \tilde{P}_{ij} \exp\big(\Theta(i) - \Theta(j) \big),
\end{equation}
with the straightforward interpretation of the information flow moving from high- to low- potential nodes.
Unlike our earlier potential \eqref{eq:potn2}, which was totally heuristic, this new potential is theoretically founded.

\section{Applications to cellular networks}\label{sec:application}

In the recent years, development of high-throughput genomic and proteomic techniques resulted in proteome-wide interaction networks (interactomes) in a number of model organisms \citep{ICOYHS01,UGCM00,GBBC03,LABG04,SWLH05,RVHH05,PDMZ05}.  Databases such as the BioGRID~\citep{BSRB08}, IntAct~\citep{KAAB07}, DIP~\citep{ SMSPBE04} and MINT~\citep{CCPN07} have been established to collect and curate sets of interactions from different experiments and make them publicly available. Most databases contain physical binding interactions, while the BioGRID additionally includes genetic interactions (such as synthetic lethality) and biochemical interactions, which describe a biochemical effect of one protein upon another. 

\begin{figure}[tbp]
\begin{center}
\includegraphics[scale=0.75]{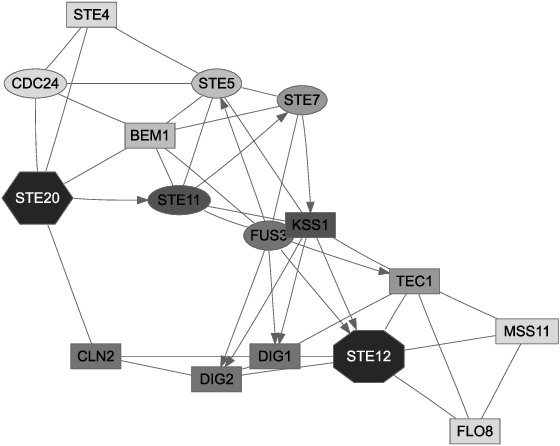}
\caption[ITMs for the MAPK cascade part of the yeast pheromone response]{ITMs for the MAPK cascade part of the yeast pheromone response obtained by running the normalized channel mode with Ste20p as the source and Ste12p as the sink ($\mu=0.85$). In addition to the `standard' excluded nodes (histones, chaperones, cytoskeleton), we also removed the nodes for Slt2p and Nup53p as discussed in the main text. Grey shading of each node indicates its total content (darker nodes represent more visits). The number of nodes shown is determined by the participation ratio.}\label{fig:fig1}
\end{center}
\end{figure}

A protein (or a protein state) is mapped to a node in a cellular protein network. Hence, the solution of a channel mode context (as tensors $\bPhiT$ and $\bnPhiT$) will highlight an ITM consisting of the proteins most visited by a directed flow from sources to sinks, that is, the proteins lying on the most likely paths connecting sources and sinks. Clearly, biological interpretations of the model results will depend on the nature of interactions ascribed  for links within the network graphs: the interpretation for an ITM from a genetic or functional network 
 and that for an ITM from a physical network should be different. Here, we will mainly focus on the physical networks where interactions correspond to binding between two proteins (undirected) or a post-translational modification of one protein by another (directed). Each step of a random walk in such a network is equivalent to a physical event and dissipation naturally corresponds to protein degradation by a protease and negative feedback mechanisms that limit transmission of information. It is thus plausible that the information channels obtained by solving the channel mode with suitable sources and sinks may correspond to (portions of) actual signaling or gene regulation pathways. However, it is important to note that the biological validity of a network path is contingent upon the transitivity of biochemical effect along that path as not all protein-protein interactions have the same downstream effect. Also, even in the best case, the information obtained from a random walk models would be primarily qualitative since cellular processes in general do not correspond to linear models.

The simplest way to use the channel mode is for knowledge retrieval by exploring large networks. In many model organisms, it is possible to construct physical protein interaction networks that integrate proteome-wide data collected from results of multiple experiments from different sources using a variety of techniques. All three modes discussed in this paper, emitting, absorbing and channel, can be used to explore network neighborhoods of proteins of interest and learn more about their function(s). In particular, given two proteins, one set as a source and the other as a sink, one may use the channel mode to extract a sub-network containing only the proteins most relevant to the possible functional relation between them. By using graphical tools to visualize the sub-network and by examining the annotations for the individual proteins within it, one can learn about their role within the cell and hence understand the cellular context of the query proteins.

More complex settings of the channel mode can be used for hypothesis forming and confirmation. For example, using destructive interference between flows from multiple sources may reveal the points of crosstalk between different biological pathways that can be selected for further experimental investigation. Given one or more proteins of interest one can explore the hypothesis about their function by using the property that sinks split the flow. Set these proteins of interest as sources and set several sinks, each associated with an a different biological role. After running a channel mode, the sinks attracting most of the flow would point to the most likely cellular role of the proteins, \emph{given all alternatives}. Of course, if all alternatives are biologically invalid, no valid functional inference can be made.

\begin{figure}[tbp]
\begin{center}
\includegraphics{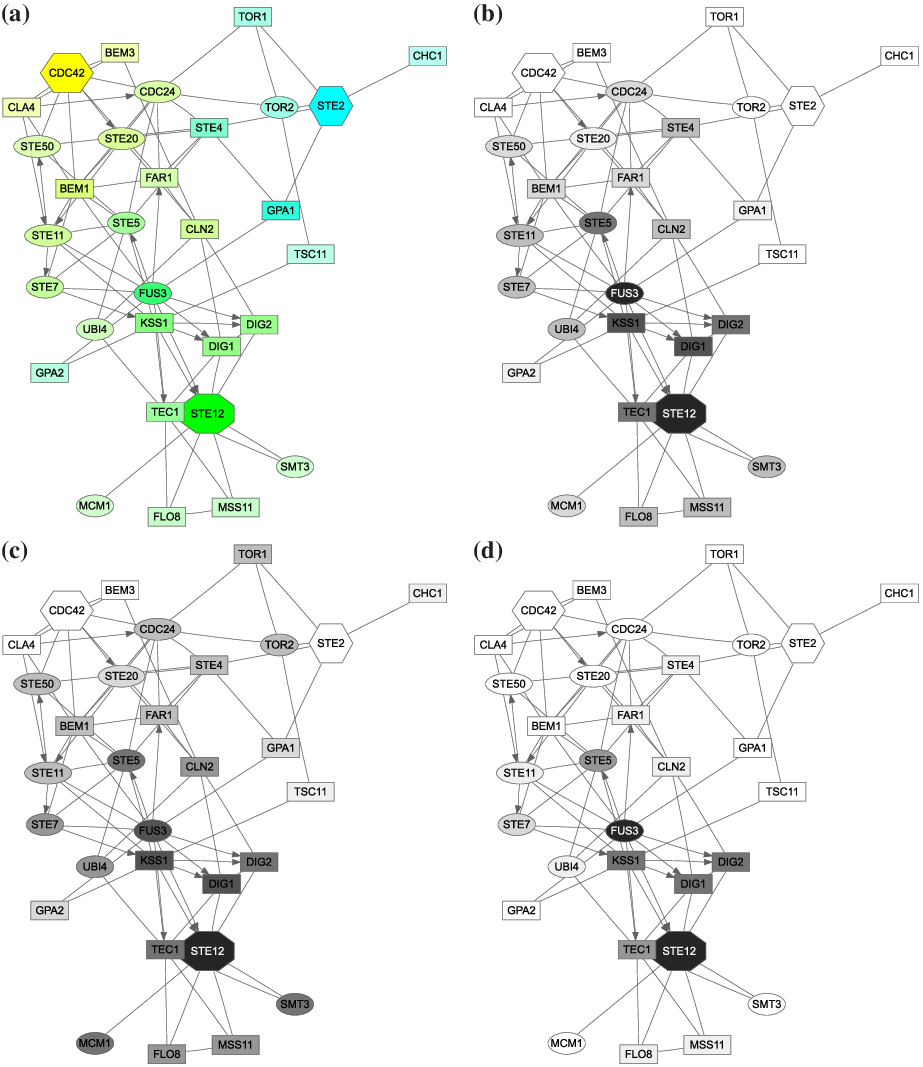}
\caption[Yeast pheromone response ITMs]{Yeast pheromone response ITMs obtained by running the normalized channel mode with Ste2p and Cdc42p as the sources and Ste12p as the sink with damping factors $\mu=0.85$ (\textbf{(a)} and \textbf{(b)}), $\mu=1$ \textbf{(c)} and $\mu=0.55$ \textbf{(d)}. Part \textbf{(a)} shows flow intensity from each source using a separate base color, while \textbf{(b)}, \textbf{(c)} and \textbf{(d)} show interference (darker nodes indicate stronger interference). All images show the top 30 nodes in terms of the total content for the case of $\mu=0.85$.}\label{fig:fig2}
\end{center}
\end{figure}

Since it is possible to arbitrarily specify sources and sinks and obtain model results that may not correspond to any cellular role, it is desirable to be able to check whether retrieved ITMs can be associated with any existing annotation. A common way to do so is through enrichment analysis~\citep{HSL09}, which assigns terms from a controlled vocabulary such as Gene Ontology~\citep{ABB00} or KEGG~\citep{KGF10} to a set of genes or proteins with weights. Each term from a controlled vocabulary annotates one or more proteins and enrichment analysis aims to retrieve, by statistical inference, those terms that best describe the set of submitted proteins with weights. While many enrichment tools were developed for analysis of microarrays~\citep{HSL09}, we found that none of them are entirely suitable for analyzing the results of our model. We have therefore developed a novel tool, called \emph{SaddleSum}~\citep{SY10a}, which is based on asymptotic approximation of tail probabilities~\citep{LR80}. For each term, it computes the probability that a score greater than or equal to the sum of weights, for all the proteins associated with that term, can arise by chance. In the context of the channel mode, the quantities that can serve as input to \emph{SaddleSum} are source specific content, total content, and destructive interference.


\subsection{Example: Yeast Pheromone Pathway}\label{subsec:pheromone}

\begin{figure}[tbp]
\begin{center}
\includegraphics{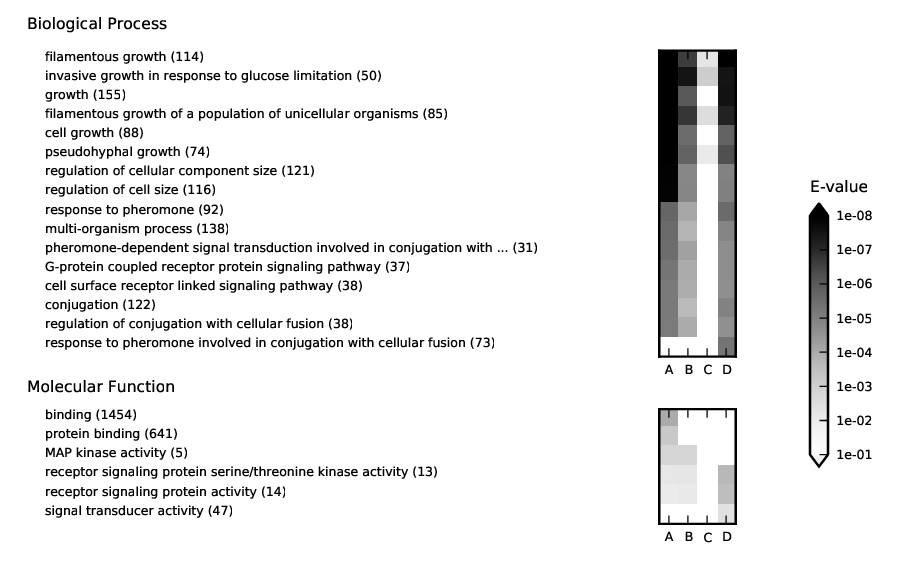}
\caption[GO term enrichment analysis]{Gene Ontology term enrichment analysis of examples from Fig.~\ref{fig:fig1} and \ref{fig:fig2} using SaddleSum.  The most significant GO terms from the Biological Process and Molecular Function categories are shown on the left (number of annotated proteins is in brackets), with their E-values indicated by shading of the squares on the right. Each column corresponds to a single example: A -- Fig.~\ref{fig:fig2}c ($\mu=1$); B -- Fig.~\ref{fig:fig2}b ($\mu=0.85$); C -- Fig.~\ref{fig:fig2}d ($\mu=0.55$); D -- Fig.~\ref{fig:fig1} ($\mu=0.85$). The input weights for columns A, B and C were obtained from the interference values at all non-excluded nodes except sources and sinks, while total content was used for column D. E-values larger than the cutoff of 0.01 are shown as white squares.}\label{fig:fig7}
\end{center}
\end{figure}

As an illustration, we will consider the mating pheromone response pathway in \textit{Saccharomyces cerevisiae}, the one of the best understood signalling pathways in eukaryotes \citep{Bardwell2005}. The mating signal is transferred from a membrane receptor to a transcription factor in nucleus, leading to transcription of mating genes. We will only provide a very brief overview of the pathway necessary for discussing our examples; more details are available in the review by \citet{Bardwell2005}.

A mating pheromone binds the transmembrane G-protein coupled pheromone receptors Ste2p/Ste3p. This leads to dissociation of Ste4p and Ste18p, the membrane bound subunits of the G-protein complex, which also contains the subunit Gpa1p. Ste4p then binds to the protein kinase Ste20p, which is recruited to the membrane through Cdc42p, and the scaffold protein Ste5p. On the scaffold, a MAPK (mitogen activated protein kinase) cascade occurs, where each protein kinase in the cascade is activated by being phosphorylated by the previous kinase and in turn activates the next protein. In this case, the cascade goes Ste20p $\to$ Ste11p $\to$ Ste7p $\to$ Fus3p or Kss1p. The final activated kinase Fus3p or Kss1p then migrates to the nucleus where it phosphorylates the proteins Dig1p and Dig2p, the repressors of the Ste12p transcription factor activity. The Ste12p transcription factor can then, in coordination with other transcription factors such as Tec1p, promote the transcription of the mating genes.

As a basis for the underlying network, we used all physical yeast protein-protein interactions from the BioGRID-3.0.65~\citep{BSRB08}. To improve the fidelity of the network, we removed every interaction reported by a single publication unless that publication described a low-throughput experiment, which we assumed to be more targeted. We considered an experiment low-throughput if it reported fewer than 300 interactions in total. We also removed all interactions labelled with the `Affinity Capture-RNA' experimental system since they were not protein-to-protein. The physical binding interactions were given a weight 1 in both directions while the interactions labelled as `Biochemical Activity' were interpreted as directional (bait $\to$ prey) and given a weight of 2. In cases where both physical and biochemical interactions were reported, only biochemical were considered. Since it is known \citep{SPAD02} that proteins with a large number of non-specific interaction partners might overtake the true signaling proteins in the information flow modeling, we excluded a set of 165 nodes corresponding to cytoskeleton proteins, histones and chaperones. We found that the excluded nodes do not strongly affect the results of the particular examples presented here. For each example we computed the normalized channel tensor summed over all sinks, that is $\nPhiT{i}{s}{K}$ in our notation.

Fig.~\ref{fig:fig1} focuses solely on the MAPK cascade portion of the pheromone pathway, with Ste20p as a single source and Ste12p as a single sink. Selection of top proteins by participation ratio captures all important participants of the cascade but emphasizes a `shortcut' through Slt2p, which is a MAP kinase involved in a different signalling pathway. Upon examination of the reference \citep{ZMM96} used by the BioGRID to support the Ste20p $\to$ Slt2p link, we discovered that it does not anywhere claim existence of such interaction. Hence, we removed Slt2p from our network for all subsequent queries and reran the query. In addition to the true pathway, the second ITM emphasized a path through Nup53p (a nuclear core protein). We examined the publication \citep{LWMD07} indicated by the BioGRID to support the Ste20p $\to$ Nup53p link and found that while it is true that Ste20p phosphorylates Nup53p \emph{in vitro}, another kinase was mainly responsible for its phosphorylation \emph{in vivo}. We therefore felt justified to exclude Nup53p as well. The final ITM resulting from the same query with Slt2p and Nup53p 
excluded in addition to the $165$ proteins mentioned before is shown in Fig.~\ref{fig:fig1}. Enrichment analysis using the GO database (Fig.~\ref{fig:fig7}, column D) gives `receptor signaling protein serine/threonine kinase activity' as a top term under `Molecular Function' and `filamentous growth' as a top term under `Biological Process'. Hence, the final ITM agrees well with the canonical understanding of the MAPK cascade.

To obtain an ITM best describing the entire pheromone response pathway, it is necessary to include two sources, the receptor Ste2p and the membrane-bound protein Cdc42p (Fig.~\ref{fig:fig2}). Including only Ste2p is not sufficient because of the shortcut through the link Gpa1p $\to$ Fus3p, which avoids the MAPK cascade. Furthermore, inclusion of Cdc42p is biologically sensible because Cdc42p activates Ste20p \citep{Bardwell2005} and is hence necessary for the MAPK cascade. Since the information flows from Ste2p and Cdc42p to Ste12p share some but definitely not all paths in common (Fig.~\ref{fig:fig2}a), interference between them (Fig.~\ref{fig:fig2}b), rather than total visits, is most appropriate to highlight the most important proteins in the signalling pathway. 

Figs.~\ref{fig:fig2} (b,c and d) illustrate the effect of changing the damping factor $\mu$. With $\mu=1$ (Fig.~\ref{fig:fig2}c) the flows from sources visit a much larger portion of the network (the average path length $\bar{T}_{sK}= \frac{1}{\abs{S}}\sum_{s\in S} T_{sK} = 19.32$) than with $\mu=0.85$  (Fig.~\ref{fig:fig2}b, $\bar{T}_{sK}=7.14$) or $\mu=0.55$  (Fig.~\ref{fig:fig2}d, $\bar{T}_{sK}=4.58$). The lower bound on path length is $3$, the shortest distance from both sources to Ste12p.  In terms of enrichment analysis with GO (Fig.~\ref{fig:fig7}, columns A,B and C), all three cases pick as significant the terms related to cell growth but with different statistical significance. In addition, both the $\mu=0.85$ and $\mu=1$ cases can be associated with terms related to MAP kinase and signal transduction. Hence, results for large $\mu$ tend to give lower GO term E-values but with lower specificity while results for small $\mu$ give very specific results but with less significant E-values. The $\mu$-dependence of E-values for any given term is not surprising since different $\mu$s correspond to different null models. Based on the images in Fig.~\ref{fig:fig2}, the enrichment results as well as our experience in other model contexts, the values of $\mu$ around 0.85, corresponding to a random walk visiting about four more nodes than the minimum necessary to reach the sink, appear to give the best results in emphasizing biologically relevant channels.

\begin{figure}[tbp]
\begin{center}
\includegraphics{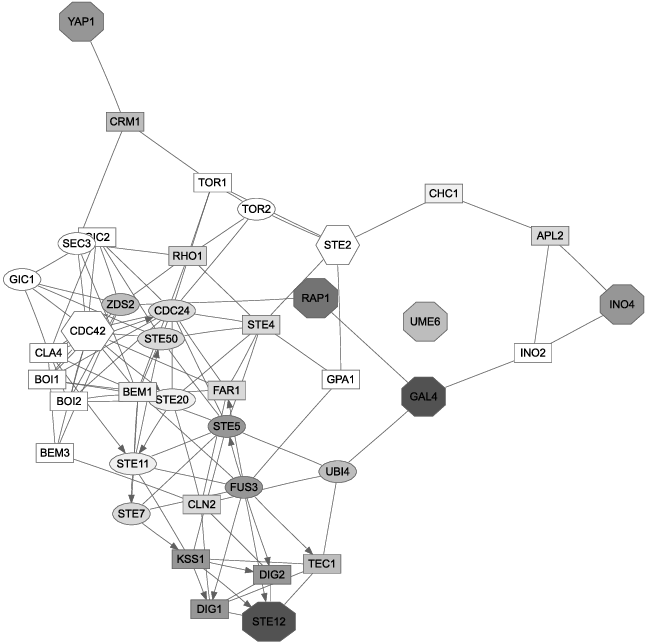}
\caption[Alternative transcription factor targets of yeast pheromone response pathway]{Alternative transcription factor targets of yeast pheromone response pathway. ITM was obtained by running the normalized channel mode with Ste2p and Cdc42p as the sources and the transcription factors Ste12p, Gal4p, Ino4p, Ume6p, Yap1p and Rap1p as the sinks with damping factor $\mu=0.85$. Nodes are shaded by interference. Most of the flow still reaches the proper target Ste12p while the channels towards other sinks are weak.}\label{fig:fig3}
\end{center}
\end{figure}

The channel mode is relatively robust to addition of non-relevant sinks to its contexts. In Fig.~\ref{fig:fig3}, we set as sinks Ste12p plus five additional transcription factor proteins not known to be directly influenced by the pheromone response pathway. As can be seen, the most visited nodes mostly belong to the channel to Ste12p while the remaining sinks are linked to sources by weaker channels (mostly not shown because the figure shows only the top 40 nodes). In this case, Ste12p has $0.62$ total visits (out of $2$) with interference of $0.54$. The remaining $1.38$ visits are distributed among the other five sinks. Enrichment results are similar to those with additional sinks absent.

\begin{figure}[tbp]
\begin{center}
\includegraphics{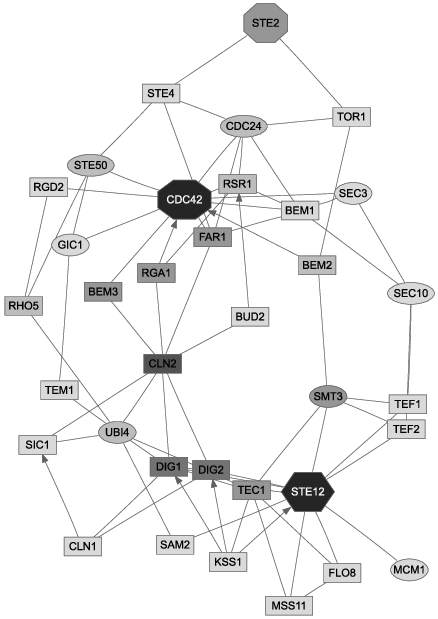}
\caption[Reversal of sources and sinks.]{Reversal of sources and sinks for the yeast pheromone response pathway. ITM was obtained by running the normalized channel mode with Ste2p and Cdc42p as the sinks and Ste12p as the source ($\mu=0.85$). Nodes are shaded by total content. The flow uses entirely different channels from Fig~\ref{fig:fig2} and the MAPK cascade is missing.}\label{fig:fig4}
\end{center}
\end{figure}

Fig.~\ref{fig:fig4} shows the effects of reversing sources and sinks. The resulting ITM performs much worse in describing the pheromone pathway for both reasons discussed in the last paragraph of \ref{subsubsec:channeltens}. Firstly, the pheromone response pathway is dominated by the MAPK phosphorylation cascade, which is in our case modelled by directed links `towards' Ste12p. Thus, the cascade cannot be seen at all in the image. Secondly, since the sinks are competing, most of the information emitted from Ste12p is captured by Cdc42p, leaving little for Ste2p.

\section{Discussion and Conclusion}\label{sec:discussion}

We have described the channel mode for modeling context-specific information flow in interaction networks. It supports discovery of the most likely channels through networks between user-specified origins (sources) and destinations (sinks) of information. The transition operator $\mat{N}$, constructed by applying potentials centered on sinks to the original transition operator, fully describes the dynamics of the flow within the channels. The mathematical formulation of the channel mode is flexible and can be easily modified for related cases. For example, it is possible to model the flow through a sequence of `checkpoints' by using destination from one context as the origin for another. 

Unlike other models based on random walks and/or electrical networks proposed in the literature \citep{TWACS06,SBKEI08,MLZR09,VTX09} that can be reduced to either emitting or absorbing modes, our channel mode allows for ``directed" information flow. Furthermore, it can readily accommodate networks containing directed links and multiple sources and sinks. Most importantly, like our original framework (absorbing and emitting modes), the channel mode uses damping to retain the information flow in the portion of the network most relevant to the specified context and prevent visits to distant nodes. Damping is controlled by  a free parameter $\mu$ (or more generally, node specific parameters $\alpha_i$), which in the case of the channel mode controls the amount of path deviation from the shortest one. In statistical physics terms, this is equivalent to using fugacity to control the number of particles of the system. Computation of the model solution requires only a solution to a (sparse) system of linear equations, without needing to simulate random walks as was done in \citep{TWACS06}. If computation of multiple contexts with the same damping coefficients is required, it is possible, using well known results from linear algebra (Appendix~\ref{app:invcompute}), to re-use the Green's function for one context to efficiently compute the Green's function for another.

Applied to physical protein interaction networks, the channel mode can be used as a framework for knowledge retrieval through network exploration and hypothesis formation and confirmation. The node weights obtained can be interpreted directly as well as submitted to an enrichment tool for further analysis. Note however that, in many cases, the annotation of a protein by a term is directly tied to publications reporting its physical interactions. 

As illustrated by our pheromone pathway example, the results of our model are sensitive to `shortcuts' between biologically unrelated protein nodes. Therefore, to obtain valid conclusions from the ITMs retrieved, the underlying interaction network must be constructed from high-quality data relevant to the biological context under investigation. The nodes with many non-specific interactions, as well as links that may not represent actual \textit{in vivo} interactions under the query context, should be removed from the network. The damping factor $\mu$ also needs to be selected appropriately for the biological context investigated, depending on whether the coverage (high $\mu$) or the selectivity (low $\mu$) of the channel are desired more. The results of enrichment analysis for the example shown in Fig.~\ref{fig:fig2} indicate that at least some interpretations are robust to the change of $\mu$.

We have already deployed a software implementation of our framework, called \ITMProbe, to the web for the use of biomedical researchers~\citep{SY09}. In future, we plan to extend our information flow framework to a platform for network-based context-specific qualitative analysis of cellular process. The improved models will take into account additional biological information, such as protein complex memberships, post-translational modification states and abundances, possibly leading to non-linear transition operators. Generally, while wishing to improve accuracy by incorporating more specific information, we aim to preserve the simplicity of the present framework.

\section*{Acknowledgments}
This work was supported by the Intramural Research Program of the National Library of Medicine at the National Institutes of Health.

\newpage

\appendix
\section*{Appendix}
\section{Channel tensor as expectation}

\begin{lemma}\label{lem:channelexpct}
Let $Z^s_{i,k}$ be a random variable denoting the total number of times a random walk starting at $s\in S$ and absorbed at $k\in K$ visits $i\in V$. Then,
\begin{equation}
\E(Z^s_{i,k}) = \PhiT{i}{s}{k}.
\end{equation}
\end{lemma}
\begin{proof}

Consider a path $x=x_0,x_1,x_2\ldots x_\tau$ from $s\in S$ to $k\in K$ of total length $\tau$ where $x_0=s$, $x_\tau=k$ and $x_1,x_2,\ldots x_{\tau-1}\in T$. The total weight or probability associated with $x$ is $\Pr(x)=P_{x_0 x_1}P_{x_1 x_2}\ldots P_{x_{\tau-1} x_{\tau}}$. For any $i\in V$, let $X_i(x,t)=1$ if $x_t=i$ and $0$ otherwise. Then, the total number of times $x$ visits $i$ is $N_i(x)=\sum_{t=0}^{\tau} X_i(x,t)$ and 
\begin{equation*}
Z^s_{i,k}=\sum_{\tau=1}^\infty \sum_{x\in\mathcal{X}(\tau)} N_i(x),
\end{equation*}
where $\mathcal{X}(\tau)$ denotes the set of all paths from $s$ to $k$ of length $\tau$. Therefore,
\begin{align}
\E(Z^s_{i,k}) = \sum_{\tau=1}^\infty \sum_{x\in\mathcal{X}(\tau)} N_i(x)\Pr(x) &= \sum_{\tau=1}^\infty \sum_{x\in\mathcal{X}(\tau)} \sum_{t=0}^{\tau} X_i(x,t)\Pr(x) = \sum_{\tau=1}^\infty \sum_{t=0}^{\tau} Y_i(t;\tau),
\end{align}
where $Y_i(t;\tau)= \sum_{x\in\mathcal{X}(\tau)} X_i(x,t)\Pr(x)$. There are three cases to consider depending on whether $i$ is a source, a sink or a transient node.

If $i$ is a source, a path from $s$ can visit $i$ only if $i=s$ and $t=0$. Therefore, $X_i(x,t)=\delta_{si}\delta_{t0}$ and hence 
\begin{equation}
Y_i(t;\tau) = 
\begin{cases} 
\delta_{si}P_{sk}&\text{if $t=0$ and $\tau=1$},\\
\sum_{j,j'\in T}\delta_{si}P_{ij}\left[\mat{P}_{TT}^{\tau-2}\right]_{jj'}P_{j'k} & \text{if $t=0$ and $\tau\geq 2$}, \\
0 & \text{otherwise}.
\end{cases}
\end{equation}
Here $\left[\mat{P}_{TT}^{\tau-2}\right]_{jj'}$ is exactly the total weight of paths of length $\tau-2$ that start at $j\in T$, visit nodes in $T$ and terminate at $j'\in T$. Hence, 
\begin{align}
\E(Z^s_{i,k}) &= \delta_{si}P_{ik} + \sum_{\tau=2}^\infty \sum_{j,j'\in T}\delta_{si}P_{ij}\left[\mat{P}_{TT}^{\tau-2}\right]_{jj'}P_{j'k} =  \delta_{si}\left[\mat{P}_{SK}\right]_{ik} +  \delta_{si}\sum_{j,j'\in T} P_{ij}\sum_{n=0}^\infty\left[\mat{P}_{TT}^n\right]_{jj'}P_{j'k}  \notag\\
&= \delta_{si}\left[\mat{P}_{SK} + \mat{P}_{ST}\mat{G}\mat{P}_{TK}\right]_{ik} = \bar{H}_{si}\bar{F}_{ik} = \PhiT{i}{s}{k}.
\end{align}

Similarly, if $i$ is a sink, a walker from $s$ can visit $i$ and terminate at $k$ only if $i=k$ and $0<t=\tau$. Thus, $X_i(x,t)=\delta_{ik}\delta_{t\tau}$ and
\begin{equation}
Y_i(t;\tau) =
\begin{cases} 
P_{si}\delta_{ik}&\text{if $t=\tau=1$},\\
\sum_{j,j'\in T}P_{sj}\left[\mat{P}_{TT}^{\tau-2}\right]_{jj'}P_{j'i}\delta_{ik} & \text{if $t=\tau\geq 2$}, \\
0 & \text{otherwise}.
\end{cases}
\end{equation}
Therefore,
\begin{align}
\E(Z^s_{i,k}) &= P_{si}\delta_{ik} + \sum_{\tau=2}^\infty \sum_{j,j'\in T}P_{sj}\left[\mat{P}_{TT}^{\tau-2}\right]_{jj'}P_{j'i}\delta_{ik} =  \left[\mat{P}_{SK}\right]_{si}\delta_{ik} +  \sum_{j,j'\in T} P_{sj}\sum_{n=0}^\infty\left[\mat{P}_{TT}^n\right]_{jj'}P_{j'i}\delta_{ik}  \notag\\
&= \left[\mat{P}_{SK} + \mat{P}_{ST}\mat{G}\mat{P}_{TK}\right]_{si}\delta_{ik}  = \bar{H}_{si}\bar{F}_{ik} = \PhiT{i}{s}{k}.
\end{align}

Finally, suppose $i\in T$. In order to visit $i$ at time $t$ and terminate at $k$ at time $\tau$, a path in $\mathcal{X}(\tau)$ must take one step to reach $T$, spend there $t-1$ steps before arriving at $i$, then take another $\tau-t-1$ steps in $T$ and an additional step to terminate at $k$. Considering all possible paths that visit $i$ at time $t$, we have  
\begin{equation}
Y_i(t;\tau) =
\begin{cases} 
\sum_{j,j'\in T}P_{sj}\left[\mat{P}_{TT}^{t-1}\right]_{ji}\left[\mat{P}_{TT}^{\tau-t-1}\right]_{ij'}P_{j'k} &\text{if $1\leq t < \tau$},\\
0 & \text{otherwise}.
\end{cases}
\end{equation}
It follows that
\begin{align}
\E(Z^s_{i,k}) &= \sum_{\tau=2}^\infty  \sum_{t=1}^{\tau-1} \sum_{j,j'\in T}P_{sj}\left[\mat{P}_{TT}^{t-1}\right]_{ji}\left[\mat{P}_{TT}^{\tau-t-1}\right]_{ij'}P_{j'k} = \sum_{t=1}^\infty\sum_{\tau=t+1}^\infty \sum_{j,j'\in T}P_{sj}\left[\mat{P}_{TT}^{t-1}\right]_{ji}\left[\mat{P}_{TT}^{\tau-t-1}\right]_{ij'}P_{j'k} \notag\\
&= \sum_{j,j'\in T}P_{sj}\sum_{n=0}^\infty\left[\mat{P}_{TT}^{n}\right]_{ji}\sum_{m=0}^\infty\left[\mat{P}_{TT}^{m}\right]_{ij'}P_{j'k} = \left[\mat{P}_{ST}\mat{G}\right]_{si}\left[\mat{G}\mat{P}_{TK}\right]_{ik}  = \bar{H}_{si}\bar{F}_{ik} = \PhiT{i}{s}{k}. \qedhere
\end{align}
\end{proof}

\section{Reversibility of sources and sinks}\label{apndx:reverse}

It is easy to see that in general, reversing sources and sinks produces different values for the normalized channel tensor. One important exception, however, is the case when the underlying graph is undirected and there is a single source and a single sink. 

\begin{lemma}\label{lem:reverse}
Let $\Gamma=(V,E,w)$ be an \emph{undirected} weighted graph with a weight matrix $\mat{W}$ and transition matrix $\mat{P}$ as defined in \eqref{eq:pmatrix}, with $\alpha_i\in[0,1]$ for all $i\in V$. Suppose $\Gamma$ is connected and let $s,k\in V$. Denote by $\bnPhiT$ the normalized channel tensor over $\Gamma$ with $s$ as a single source and $k$ as a single sink, and denote by $\bnPsiT$ the normalized channel tensor over $\Gamma$ with $k$ as a single source and $s$ as a single sink. Then, for all $i\in V$,
\begin{equation}
\nPhiT{i}{s}{k} = \nPsiT{i}{k}{s}.
\end{equation}
\end{lemma}
\begin{proof}
Since $\Gamma$ is an undirected graph, it satisfies the detailed balance equation $\pi_y P_{xy} = P_{yx} \pi_x$
for all $x,y\in V$, where $\pi_x = \alpha_x / \sum_{z\in V} W_{xz}$.  It directly follows that 
\begin{equation}
\pi_y G_{xy} = \sum_{n=0}^\infty \pi_y [\mat{P}^n_{TT}]_{xy} = \sum_{n=0}^\infty [\mat{P}^n_{TT}]_{yx}\pi_x  =  G_{yx} \pi_x. 
\end{equation}
For $i=s$ or $i=k$, one can immediately see that $\nPhiT{i}{s}{k} = 1 = \nPsiT{i}{k}{s}$. Observing that the transient state is the same for both $\bnPhiT$ and $\bnPsiT$, we have for each $i\in T$, 
\begin{align*}
\nPhiT{i}{s}{k} = \frac{\left(\sum_{j\in T} P_{sj}G_{ji} \right) \left( \sum_{j'\in T} G_{ij'}P_{j'k} \right) }{P_{sk} + \sum_{j,j'\in T} P_{sj}G_{jj'}P_{j'k}} = \frac{\left(\sum_{j\in T} \frac{\pi_s}{\pi_j} P_{js} \frac{\pi_j}{\pi_i} G_{ij} \right) \left( \sum_{j'\in T} \frac{\pi_i}{\pi_{j'}}  G_{j'i} \frac{\pi_{j'}}{\pi_k}  P_{kj'} \right) }{ \frac{\pi_{s}}{\pi_k} P_{ks} + \sum_{j,j'\in T} \frac{\pi_s}{\pi_j} P_{js} \frac{\pi_j}{\pi_{j'}} G_{j'j} \frac{\pi_{j'}}{\pi_k} P_{kj'}} = \nPsiT{i}{k}{s}.
\end{align*}
\end{proof}

\section{The role of the damping factor in the channel mode}\label{apndx:A}

Recall that $\mat{P}=\mu\mat{Q}$, where $\mu\in (0,1)$ is the uniform damping factor and $\mat{Q}$ is given in \eqref{eq:pmatrix1}. For this range of $\mu$,
the Green's function $\mat{G}=(\I-\mat{P}_{TT})^{-1}=\sum_{n=0}^\infty \mat{P}_{TT}^n=\sum_{n=0}^\infty \mat{Q}_{TT}^n\mu^n$ is well-defined (see \citep{SY07}, Proposition 2.2) and hence the solution matrices $\mat{\bar{F}}$ and $\mat{\bar{H}}$ from Equations (\ref{eq:barF}--\ref{eq:barH}) are well defined and continuous as functions of $\mu$. As $\mu\downarrow 0$, all the damping factors in $\balpha$ uniformly tend to $0$ and $\mat{P}\to \mat{0}$.
However, we will show in \ref{subsec:damping0} that the normalized channel tensor is well-defined in the limit as $\mu\downarrow 0$ (provided it is well defined for other values of $\mu$). 

At the other extreme, as $\mu\uparrow 1$ and $\mat{P}\to \mat{Q}$, the Green's function may not exist for every choice of boundary nodes, since the spectral radius of $\mat{Q}_{TT}$ may be equal to $1$. If the vertex set is restricted to $V(K)$, the set of all nodes connected through a directed path to at least one sink, then by Proposition 2.1 of \citep{SY07}, the Green's function is well-defined for $\mu=1$ as well. Also note that for a channel tensor $\bPhiT$ to be non-trivial (i.e. non-zero everywhere), it is also necessary that each source is connected to at least one sink through a directed path, or equivalently, that $F_{sK} > 0$ for all $s\in S$.

\subsection{Path lengths}

The damping parameter $\mu$ controls the distribution of lengths of the paths (or the times) a random walk emitted from a source takes before being absorbed at a sink.

For nodes $s\in S$ and $k\in K$, let $L_{sk}$ (more precisely, $L_{sk}(\mu)$) denote the random variable giving the length of the path (or a number of steps) taken by a random walk originating at $s$ and terminating at $k$. At least one such path from $s$ to $k$ exists if and only if $F_{sk} > 0$. The underlying probability density $\Pr(L_{sk}=n)$ is given by
\begin{equation}
\Pr(n) = \frac{1}{F_{sk}} \begin{cases}
P_{sk} & \text{for $n=1$;}\\
\left[\mat{P}_{ST}\mat{P}^{n-2}_{TT}\mat{P}_{TK} \right]_{sk} & \text{for $n\geq 2$.}
\end{cases}
\end{equation}
Let $M_{L_{sk}(\mu)}$ denote the moment generating function for $L_{sk}$ and let $C_{L_{sk}(\mu)}\equiv \log M_{L_{sk}(\mu)}$ denote its cumulant generating function. Let us write $F_{sk}$ as a function of $\mu$:
\begin{align}
F_{sk}(\mu)  
& = Q_{sk}\mu + \sum_{n=2}^\infty \left[\mat{Q}_{ST}\mat{Q}^{n-2}_{TT}\mat{Q}_{TK} \right]_{sk}\mu^n,
\end{align}
and observe that
\begin{align}
M_{L_{sk}(\mu)}(t) & = \sum_{n=0}^\infty \Pr(n)e^{nt} = P_{sk}e^{t} + \sum_{n=2}^\infty \left[\mat{P}_{ST}\mat{P}^{n-2}_{TT}\mat{P}_{TK} \right]_{sk} e^{nt} \notag \\
& = Q_{sk}\mu e^{t} + \sum_{n=2}^\infty \left[\mat{Q}_{ST}\mat{Q}^{n-2}_{TT}\mat{Q}_{TK} \right]_{sk}\mu^n e^{nt} = F_{sk}(\mu e^t).
\end{align}
Thus, all moments and cumulants of $L_{sk}$ can be expressed in terms of the Green's function $\mat{G}$ and its related quantities $\mat{F}$, $\mat{H}$ and $\bPhiT$, both directly and in terms of derivatives of their entires with respect to $\mu$. In particular, 
\begin{align}
\E(L_{sk}) & = C'_{L_{sk}(\mu)}(0) 
= \frac{\frac{\partial}{\partial t} F_{sk}(\mu e^t) } {F_{sk}(\mu e^t)}\Big\vert_{t=0} 
= \frac{\mu e^t F'_{sk}(\mu e^t)} {F_{sk}(\mu e^t)}\Big\vert_{t=0}
= \frac{\mu F'_{sk}(\mu)} {F_{sk}(\mu)}. \label{eq:mean4}
\end{align}
Using the easily provable identity $\sum_{n=0}^\infty (n+2) \mat{P}_{TT}^n = \mat{G}^2 + \mat{G}$, we have
\begin{align}
F'_{sk}(\mu) & = Q_{sk} + \sum_{n=2}^\infty \left[\mat{Q}_{ST}\mat{Q}^{n-2}_{TT}\mat{Q}_{TK} \right]_{sk} n\mu^{n-1}  = \frac{1}{\mu}\left( P_{sk} +  \sum_{n=0}^\infty (n+2) \left[\mat{P}_{ST}\mat{P}^{n}_{TT}\mat{P}_{TK} \right]_{sk} \right) \notag \\
& = \frac{1}{\mu} \left( P_{sk} + \left[\mat{P}_{ST}(\mat{G}+\mat{G}^2)\mat{P}_{TK} \right]_{sk} \right) = \frac{1}{\mu} \left( F_{sk} + \left[\mat{P}_{ST}\mat{G}^2\mat{P}_{TK} \right]_{sk} \right). \label{eq:Fprime}
\end{align}
Therefore, by \eqref{eq:mean4},
\begin{align}
\E(L_{sk}) = 1 + \frac{ \left[\mat{P}_{ST}\mat{G}^2\mat{P}_{TK}\right]_{sk}}{F_{sk}}  = 1 + \sum_{i\in T} \frac{H_{si}F_{ik}}{F_{sk}}  = 1 + \sum_{i\in T}\frac{\PhiT{i}{s}{k}}{F_{sk}}, 
\end{align}
and we obtain the following
\begin{lemma}\label{lem:meanpathlen1}
Let $s\in S$, let $k\in K$ and let $\mu\in (0,1)$. Suppose $F_{sk} > 0$. Then,
\begin{equation}
T_{sk}=\E(L_{sk}) = 1 + \sum_{i\in T}\frac{\PhiT{i}{s}{k}}{F_{sk}} = \frac{\mu}{F_{sk}}\frac{\partial F_{sk}}{\partial\mu}.
\end{equation}
\end{lemma}

Similarly, 
\begin{align}
\Var(L_{sk}) & = C''_{L_{sk}(\mu)}(0)  = \frac{\partial}{\partial t}\frac{\mu e^t F'_{sk}(\mu e^t)} {F_{sk}(\mu e^t)}\Big\vert_{t=0} \notag \\
& = \frac{\mu e^t F'_{sk}(\mu e^t) + \mu^2 e^{2t} F''_{sk}(\mu e^t)} {F_{sk}(\mu e^t)}  - \left(\frac{\mu e^t F'_{sk}(\mu e^t)} {F_{sk}(\mu e^t)} \right)^2 \Big\vert_{t=0} \notag \\
& =  \E(L_{sk}) + \frac{\mu^2 F''_{sk}(\mu)} {F_{sk}(\mu)} - \E^2(L_{sk}).
\end{align}
Using another easily provable identity $\sum_{n=0}^\infty (n+2)(n+1) \mat{P}_{TT}^n = 2\mat{G}^3$,
and Equation \eqref{eq:Fprime}, we have
\begin{align}
F''_{sk}(\mu) & = \sum_{n=2}^\infty \left[\mat{Q}_{ST}\mat{Q}^{n-2}_{TT}\mat{Q}_{TK} \right]_{sk} n(n-1)\mu^{n-2} 
= \frac{1}{\mu^2} \sum_{n=0}^\infty (n+2)(n+1) \left[\mat{P}_{ST}\mat{P}^{n}_{TT}\mat{P}_{TK} \right]_{sk} \notag \\
& = \frac{2}{\mu^2} \left[\mat{P}_{ST}\mat{G}^3\mat{P}_{TK} \right]_{sk}. 
\end{align}
Hence, we obtain
\begin{lemma}\label{lem:varpathlen}
Let $s\in S$, let $k\in K$ and let $\mu\in (0,1)$. Suppose $F_{sk} > 0$. Then,
\begin{equation}
\Var(L_{sk}) = \E(L_{sk}) + \frac{2\left[\mat{P}_{ST}\mat{G}^3\mat{P}_{TK} \right]_{sk}}{F_{sk}} - \E^2(L_{sk}).
\end{equation}
\end{lemma}

Denote by $L_{sK}$ the random variable giving the length of the path (or the number of steps) taken by a random walk originating at $s$ and terminating at any sink in $K$. This random variable is well-defined provided $s$ is connected with at least one $k\in K$ through a directed path, or equivalently, if $\max_{k\in K} F_{sk} > 0$. Let $\hat{K}(s) = \{k\in K: F_{sk} > 0\}$. Then, $L_{sK}$ can be expressed as a weighted sum of $L_{sk}$ over $k\in \hat{K}(s)$:
\begin{equation}
L_{sK} = \sum_{k\in \hat{K}(s)} \frac{F_{sk}}{F_{sK}}L_{sk}.
\end{equation}
Here $F_{sk}/F_{sK}$ gives the conditional probability of a random walker from $s$ reaching sink $k$, given that it reaches any of the sinks in $\hat{K}(s)$. Through properties of mean, we the following corollary.
\begin{corol}\label{cor:meanvar2}
Let $s\in S$ and let $\mu\in (0,1)$. Suppose $\max_{k\in K} F_{sk} > 0$. Then,
\begin{equation}
T_{sK} = \E(L_{sK}) = 1 + \sum_{i\in T}\nPhiT{i}{s}{K} = \frac{\mu}{F_{sK}}\frac{\partial F_{sK}}{\partial\mu} . \label{eq:normmean1} \\
\end{equation}
\end{corol}

Since $\E(L_{sk})$ and $\E(L_{sK})$ can be expressed in terms of sums and products of entries of $\mat{G}$, they are continuous and increasing functions of $\mu\in (0,1)$. The value of $\E(L_{sK})$ is bounded from below:
 the length of the shortest path from the source to any of the sinks. If the graph nodes are restricted to $V(K)$, $\mat{G}$ is well-defined for $\mu=1$ and $\E(L_{sK})$ is bounded and attains its maximum there. The value of the maximum varies depending on the underlying network graph and the particular context.

\subsection{Large dissipation asymptotics}\label{subsec:damping0}

For all $i,j\in V$, let $\rho(i,j)$ denote the (unweighted) length of the shortest directed path between $i$ and $j$. We allow $\rho(i,j)=\infty$ if there exists no directed path between $i$ and $j$. It is well-known that $\rho$ is a (not necessarily symmetric) distance that satisfies the triangle inequality, that is, for all $i,j,k\in V$, $\rho(i,j)+\rho(j,k)\geq \rho(i,k)$. For any source $s\in S$, recall that $\rho(s,K)=\min_{k\in K} \rho(s,k)$ and let $K_s=\{k\in K: \rho(s,k)=\rho(s,K) \}$, the set of all the sinks closest to $s$.

\begin{thm}\label{thm:nodes1}
Let $s\in S$, $i\in T$ and $k\in K$ such that $\rho(s,i)$ and $\rho(i,k)$ are both finite. Then, if $k\in K_s$ and $i$ lies on the shortest path from $s$ to $k$,
\begin{equation}\label{eq:nPhiTlim}
 \lim_{\mu\downarrow 0} \nPhiT{i}{s}{k} = \frac{ \left[\mat{Q}_{ST}\mat{Q}^{\rho(s,i)-1}_{TT} \right]_{si} \left[\mat{Q}^{\rho(i,k)-1}_{TT}\mat{Q}_{TK} \right]_{ik} }{ \sum_{k'\in K_s}\left[\mat{Q}_{ST}\mat{Q}^{\rho(s,k)-2}_{TT}\mat{Q}_{TK}\right]_{sk'}}.
\end{equation}
Otherwise, $\lim_{\mu\downarrow 0} \nPhiT{i}{s}{k} = 0$.
\end{thm}

\begin{proof}
Let $s\in S$, $i\in T$ and $k\in K$. Since, $\rho(s,i)$ and $\rho(i,k)$ are finite, it follows that $\rho(s,k)$ is also finite, that is, $k$ is reachable from $s$ through $i$ and the normalized channel tensor $\bnPhiT$ is well defined for all $\mu\in(0,1)$. Recall that
\begin{equation}\label{eq:nPhiT2}
\nPhiT{i}{s}{k} = \frac{\PhiT{i}{s}{k}}{F_{sK}} = \frac{[\mat{P}_{ST}\mat{G}]_{si}[\mat{G}\mat{P}_{TK}]_{ik}}{\sum_{k'\in K}F_{sk'}}
\end{equation}
where $F_{sk'} = [\mat{P}_{SK}+ \mat{P}_{ST}\mat{G}\mat{P}_{TK}]_{sk'}$.

Let $u,v\in T$ and let $d=\rho(u,v)$. It can be easily shown (see Lemma A.3 from \citep{SY07} for a partial proof) that $\left[\mat{P}_{TT}^n\right]_{uv}=0$ for all $n<d$ and that $\left[\mat{P}_{TT}^{d}\right]_{uv}>0$. Therefore, 
\begin{equation*}
G_{uv} =  \sum_{n=d}^\infty \left[\mat{P}_{TT}^n \right]_{uv} = \sum_{n=d}^\infty \mu^n\left[\mat{Q}_{TT}^n \right]_{uv} = \mu^{d} \left[\mat{Q}^{d}_{TT}\right]_{uv} + O(\mu^{d+1})
\end{equation*}
as $\mu\downarrow 0$. Hence,
\begin{align}
[\mat{P}_{ST}\mat{G}]_{si} &= \sum_{j\in T} \mu^{\rho(j,i)+1} Q_{sj}\left[\mat{Q}^{\rho(j,i)}_{TT} \right]_{ji} + O(\mu^{\rho(j,i)+2}) = \mu^{\rho(s,i)} \left[\mat{Q}_{ST}\mat{Q}^{\rho(s,i)-1}_{TT} \right]_{si} + O(\mu^{\rho(s,i)+1}) ,\label{eq:limitH}\\
[\mat{G}\mat{P}_{TK}]_{ik} &= \sum_{j\in T} \mu^{\rho(i,j)+1} \left[\mat{Q}^{\rho(i,j)}_{TT} \right]_{ij}Q_{jk} + O(\mu^{\rho(i,j)+2}) = \mu^{\rho(i,k)} \left[\mat{Q}^{\rho(i,k)-1}_{TT}\mat{Q}_{TK} \right]_{ik} + O(\mu^{\rho(i,k)+1}) \label{eq:limitF}.
\end{align}
Let $\xi=\rho(s,k'')$, where $k''\in K_s$. We will consider the denominator of Equation (\ref{eq:nPhiT2}) under two separate cases, $\xi=1$ and $\xi>1$. 

If $\xi>1$, for all $k'\in K$, the vertices $s$ and $k'$ are not adjacent and thus $P_{sk'}=0$. Hence, since $s$ and $k'$ are connected, there exist $j,j'\in T$ such that $\rho(s,k')=\rho(s,j)+\rho(j,j')+\rho(j',k')=\rho(j,j')+2$, implying
\begin{align}
[\mat{P}_{ST}\mat{G}\mat{P}_{TK}]_{sk'} &= \sum_{j,j'\in T} \mu^{\rho(j,j')+2} Q_{sj} \left[ \mat{Q}^{\rho(j,j')}_{TT} \right]_{jj'}  Q_{j'k'} + O(\mu^{\rho(j,j')+3}) \notag \\
& = \mu^{\rho(s,k')} \left[\mat{Q}_{ST}\mat{Q}^{\rho(s,k')-2}_{TT}\mat{Q}_{TK}\right]_{sk'} + O(\mu^{\rho(s,k')+1}).  \label{eq:Fsk1}
\end{align}
Similarly, 
\begin{align}
F_{sK} & = \sum_{k'\in K_s} \mu^{\xi} \left[\mat{Q}_{ST}\mat{Q}^{\xi-2}_{TT}\mat{Q}_{TK}\right]_{sk'} + O(\mu^{\xi+1}),  \label{eq:Fsk2}
\end{align}
and, as $\mu\downarrow 0$,
\begin{equation}
\nPhiT{i}{s}{k} \to \frac{ \mu^{\rho(s,i)+\rho(i,k)}  \left[\mat{Q}_{ST}\mat{Q}^{\rho(s,i)-1}_{TT} \right]_{si} \left[\mat{Q}^{\rho(i,k)-1}_{TT}\mat{Q}_{TK} \right]_{ik} }{ \mu^\xi\sum_{k'\in K_s}\left[\mat{Q}_{ST}\mat{Q}^{\xi-2}_{TT}\mat{Q}_{TK}\right]_{sk'}} 
\end{equation}

By the triangle inequality and our assumptions on $s$, $i$ and $k$, 
\begin{equation}
\rho(s,i)+\rho(i,k)\geq \rho(s,k)\geq \xi.
\end{equation}
The first inequality becomes an equality if and only if $i$ lies on the shortest path between $s$ and $k$ while the second is an equality if and only if $k\in K_s$. Therefore, if the assumption of the theorem is satisfied, the value of $\nPhiT{i}{s}{k}$ converges to the value of the right hand side of Equation (\ref{eq:nPhiTlim}), while otherwise $\lim_{\mu\downarrow 0} \nPhiT{i}{s}{k}=0$.

On the other hand, if $\xi=1$, $F_{sK}\to \sum_{k'\in K_s} \mu Q_{sk'}+O(\mu^{2})$ and therefore, since $\rho(s,i)+\rho(i,k)\geq 2$,  $\nPhiT{i}{s}{k}\to 0$ as $\mu\downarrow 0$. 
\end{proof}
We have therefore shown that, as $\mu\downarrow 0$, only the nodes associated with the shortest path from each source to the sink(s) closest to it will have positive values of the normalized channel tensor -- all other entries will be exactly $0$.

\begin{corol}
Let $s\in S$ and suppose the normalized channel tensor $\bnPhiT$ is well defined for all $\mu\in (0,1)$. Then,
\begin{equation}
\lim_{\mu\downarrow 0} \E(L_{sK}) = \rho(s,k), 
\end{equation}
where $k\in K_s$.
\end{corol}

\begin{proof}
Let $s\in S$, let $k\in K_s$ and let $d=\rho(s,k)$. For $m=1,2\ldots d-1$, let $\Pi_s(m) = \{i\in T: \rho(s,i)=m\, \text{and}\, \rho(s,i)+\rho(i,k)=d\}$. The set $\Pi_s(m)$ consists of all transient nodes that are at the distance $m$ from $s$ on a shortest path from $s$ to any of the sinks closest to $s$. By Theorem \ref{thm:nodes1},
\begin{align*}
\lim_{\mu\downarrow 0}\sum_{k''\in K}\sum_{i\in T} \nPhiT{i}{s}{k''} & = \sum_{k''\in K_s}\sum_{m=1}^{d-1}\sum_{i\in \Pi_s(m)} \frac{ \left[\mat{Q}_{ST}\mat{Q}^{m-1}_{TT} \right]_{si} \left[\mat{Q}^{d-m-1}_{TT}\mat{Q}_{TK} \right]_{ik''} }{ \sum_{k'\in K_s}\left[\mat{Q}_{ST}\mat{Q}^{\rho(s,k)-2}_{TT}\mat{Q}_{TK}\right]_{sk'}} \\
& = \sum_{k''\in K_s}\sum_{m=1}^{d-1}\sum_{i\in T} \frac{ \left[\mat{Q}_{ST}\mat{Q}^{m-1}_{TT} \right]_{si} \left[\mat{Q}^{d-m-1}_{TT}\mat{Q}_{TK} \right]_{ik''} }{ \sum_{k'\in K_s}\left[\mat{Q}_{ST}\mat{Q}^{d-2}_{TT}\mat{Q}_{TK}\right]_{sk'}} \\
& = \sum_{m=1}^{d-1} \frac{ \sum_{k''\in K_s} \left[\mat{Q}_{ST}\mat{Q}^{d-2}_{TT}\mat{Q}_{TK}\right]_{sk''} }{ \sum_{k'\in K_s}\left[\mat{Q}_{ST}\mat{Q}^{d-2}_{TT}\mat{Q}_{TK}\right]_{sk'}}\\
& = d-1.
\end{align*}
Therefore, by Equation \eqref{eq:normmean1}, 
\[ \lim_{\mu\downarrow 0} \E(L_{sK}) = 1 + \lim_{\mu\downarrow 0}\sum_{k'\in K}\sum_{i\in T} \nPhiT{i}{s}{k'} = \rho(s,k).\]
\end{proof}

\section{Normalized evolution operator}\label{app:normop}

To make our arguments more concise we will here additionally assume, without loss of generality, that every node is connected to a sink via a directed path, that is, that $V_K=V$. 

Note that $\mat{N}$ is indeed well defined in the limit as $\mu\downarrow 0$. For example, if $i,j\in T$, we have from \eqref{eq:limitF}
\begin{align}
N_{ij} & = \frac{\tilde{P}_{ij} \sum_{k\in K} [\mat{G}\mat{P}_{TK}]_{jk} f_k}{ \sum_{k'\in K} [\mat{G}\mat{P}_{TK}]_{ik'} f_{k'}}  \to \frac{  \sum_{k\in K} \mu^{\rho(j,k)+1} Q_{ij} \left[\mat{Q}^{\rho(j,k)-1}_{TT}\mat{Q}_{TK} \right]_{jk} f_k }{ \sum_{k'\in K} \mu^{\rho(i,k')}  \left[\mat{Q}^{\rho(i,k')-1}_{TT}\mat{Q}_{TK} \right]_{ik'} f_{k'}} \notag \\
& = \frac{ \sum_{k\in K} \mu^{\rho(j,k)+1} Q_{ij} \left[\mat{Q}^{\rho(j,k)-1}_{TT}\mat{Q}_{TK} \right]_{jk} f_k }{\mu^{\rho(i,K)} \sum_{k'\in K} \mu^{\rho(i,k')-\rho(i,K)}  \left[\mat{Q}^{\rho(i,k')-1}_{TT}\mat{Q}_{TK} \right]_{ik'} f_{k'}} \notag \\
& =  
\begin{cases}
0 & \text{if $\rho(j,K) > \rho(i,K) - 1$,} \\
\frac{ \sum_{k\in K} Q_{ij} \left[\mat{Q}^{\rho(i,K)-2}_{TT}\mat{Q}_{TK} \right]_{jk} f_k }{\sum_{k'\in K} \left[\mat{Q}^{\rho(i,K)-1}_{TT}\mat{Q}_{TK} \right]_{ik'} f_{k'}} & \text{if $\rho(j,K) = \rho(i,K) - 1$}.
\end{cases}
\end{align}
Similar well defined limits for $N_{ij}$, with $i,j\in V$, can also be easily demonstrated using the results from Appendix~\ref{subsec:damping0}.

\begin{prop}\label{prop:normop2a}
Let $\vec{f}$ denote an arbitrary vector over $V$.
Suppose $i\in S\cup T$. Then,
\begin{equation}\label{eq:equiv}
\sum_{j\in V} N_{ij} = 1 \iff f_i = \sum_{k\in K} \bar{F}_{ik} f_k.
\end{equation}
\end{prop}
\begin{proof}
Write the vector $\vec{f}$ as $\vec{f} = [\vec{f}_{S}, \vec{f}_{T}, \vec{f}_K]^T$ and the matrix $\bar{\mat{F}}$ as $\bar{\mat{F}} = \left[\bar{\mat{F}}_{SK}, \bar{\mat{F}}_{TK}, \bar{\mat{F}}_{KK} \right]$, where $\bar{\mat{F}}_{SK}=\mat{P}_{ST}\mat{G}\mat{P}_{TK} + \mat{P}_{SK}$, $\bar{\mat{F}}_{TK}=\mat{G}\mat{P}_{TK}$ and $\bar{\mat{F}}_{KK}=\I$. The right equality from \eqref{eq:equiv} can then be written in the block matrix form as $\vec{f}_T = \bar{\mat{F}}_{TK}\vec{f}_K$, and $\vec{f}_S = \bar{\mat{F}}_{SK}\vec{f}_K$.

By definition of $\mat{N}$, our premise $\sum_{j\in V} N_{ij} = 1$ is equivalent to 
\begin{equation}\label{eq:equilibrium1}
f_i = \sum_{j\in T} P_{ij}f_j + \sum_{k\in K} P_{ik}f_k. 
\end{equation}
For $i\in T$, Equation \eqref{eq:equilibrium1} can be expressed in matrix form as $\vec{f}_T = \mat{P}_{TT}\vec{f}_T + \mat{P}_{TK}\vec{f}_K$, that is,
$(\I-\mat{P}_{TT})\vec{f}_T = \mat{P}_{TK}\vec{f}_K$. Since the matrix $\I-\mat{P}_{TT}$ is invertible by our assumption of connectivity, this is further equivalent to  
\begin{equation}\label{eq:equilibrium3}
\vec{f}_T = \mat{G}\mat{P}_{TK}\vec{f}_K = \bar{\mat{F}}_{TK}\vec{f}_K.
\end{equation}
For $i\in S$, Equation \eqref{eq:equilibrium1} can be written as $\vec{f}_S = \mat{P}_{ST}\vec{f}_T + \mat{P}_{SK}\vec{f}_K$, which using \eqref{eq:equilibrium3} is equivalent to
\begin{equation}
\vec{f}_S = \mat{P}_{ST}\mat{G}\mat{P}_{TK}\vec{f}_K + \mat{P}_{SK}\vec{f}_K = \bar{\mat{F}}_{SK}\vec{f}_K,
\end{equation}
as required.
\end{proof}

\paragraph*{Proof of Proposition \ref{prop:normop}}
\begin{proof} All properties follow from the fact that the transformation from $\tilde{\mat{P}}$ to $\mat{N}$ is a similarity transformation. \\
\textit{(i)} Let $i,j\in T$. We have
\begin{align*}
[\mat{G}(\mat{N})]_{ij} = \sum_{n=0}^\infty [\mat{N}_{TT}^n]_{ij} = \sum_{n=0}^\infty \frac{[\mat{P}_{TT}^n]_{ij}f_j}{f_i} = \frac{G_{ij} f_j}{f_i}.
\end{align*}
\textit{(ii)} Let $k\in K$ and suppose $i\in K$. Then $[\mat{\bar{F}}(\mat{N})]_{ik}=\delta_{ik}= \frac{\delta_{ik}f_k}{f_i}=\frac{\bar{F}_{ik}f_k}{f_i}$. Now suppose $i\in T$. Then, 
\begin{align*}
[\mat{\bar{F}}(\mat{N})]_{ik}= [\mat{G}(\mat{N})\mat{N}_{TK}]_{ik} = \sum_{j\in T} \frac{G_{ij}f_j}{f_i} \frac{P_{jk}f_k}{f_j} = \frac{\bar{F}_{ik}f_k}{f_i}.
\end{align*}
If $i\in S$, we have
\begin{align*}
[\mat{\bar{F}}(\mat{N})]_{ik}= [\mat{N}_{SK} + \mat{N}_{ST}\mat{G}(\mat{N})\mat{N}_{TK}]_{ik} = \frac{P_{ik}f_k}{f_i} + \sum_{j\in T}\sum_{l\in T} \frac{P_{ij}f_j}{f_i} \frac{G_{jl}P_{lk}f_k}{f_j} = \frac{\bar{F}_{ik}f_k}{f_i}.
\end{align*}
\textit{(iii)} Let $s\in S$ and suppose $i\in S$. Then  $[\mat{\bar{H}}(\mat{N})]_{si}=\delta_{si}= \frac{\delta_{si}f_i}{f_s} = \frac{\bar{H}_{si}f_i}{f_s}$. Now suppose $i\in K$. Then $[\mat{\bar{H}}(\mat{N})]_{si}=[\mat{\bar{F}}(\mat{N})]_{si} = \frac{\bar{F}_{si}f_i}{f_s} = \frac{\bar{H}_{si}f_i}{f_s}$. If $i\in T$, 
\begin{align*}
[\mat{\bar{H}}(\mat{N})]_{si} = [\mat{N}_{ST}\mat{G}(\mat{N})]_{si} = \sum_{j\in T} \frac{P_{sj}f_j}{f_s} \frac{G_{ji}f_i}{f_j} = \frac{\bar{H}_{si}f_i}{f_s}.\end{align*}
\textit{(iv)} Let $s\in S$, $i\in V$ and $k\in K$. Then,
\begin{align*}
[\bPhiT(\mat{N})]^{s}_{i,k} & = [\mat{\bar{H}}(\mat{N})]_{si}[\mat{\bar{F}}(\mat{N})]_{ik} =  \frac{\bar{H}_{si}f_i}{f_s} \frac{\bar{F}_{ik}f_k}{f_i} = \PhiT{i}{s}{k}\frac{f_k}{f_s}.
\end{align*}
\end{proof}

\section{Rapid Evaluation of Submatrix Inverses}\label{app:invcompute}

Consider an invertible block matrix $\mat{M}=\left[ \begin{array}{cc}\mat{A} & \mat{B}\\ \mat{C} & \mat{D}\end{array}\right]$, where $\mat{A}$ is a square matrix. It is a well known result of linear algebra (see for example \citet{PTVF07}, 2.7.4) that the inverse of $\mat{M}$ can be written as
\begin{equation}\label{eq:blockinv}
 \mat{M}^{-1} = \left[ \begin{array}{cc}\mat{A}^{-1}+\mat{A}^{-1}\mat{B}\mat{Q}^{-1}\mat{C}\mat{A}^{-1} & -\mat{A}^{-1}\mat{B}\mat{Q}^{-1}\\ -\mat{Q}^{-1}\mat{C}\mat{A}^{-1} & \mat{Q}^{-1}\end{array}\right],
\end{equation}
where $\mat{Q} = \mat{D}-\mat{C}\mat{A}^{-1}\mat{B}$. Suppose we are interested in computing matrices of the form $\mat{A}^{-1}\mat{U}$, where $\mat{A}$ is very large and $\mat{U}$ is an arbitrary matrix with appropriate number of rows. If it is necessary to perform a large number of such computations with different square submatrices $\mat{A}$ (the matrix $\mat{M}$ may be permuted in each case to reorder the indices), it could be effective to precompute the matrix $\mat{M}^{-1}$ (or, computationally more appropriately, its LU-decomposition) once and in each case extract the required inverse $\mat{A}^{-1}$ through simple and relatively inexpensive algebraic manipulations and permutations.

Indeed, write $\mat{M}^{-1}=\left[ \begin{array}{cc}\mat{X} & \mat{Y}\\ \mat{Z} & \mat{W}\end{array}\right]$, with each of the blocks known and with the block sizes the same as that in Equation (\ref{eq:blockinv}). One observes that $\mat{W}=\mat{Q}^{-1}$ and hence $\mat{Y}\mat{W}^{-1}\mat{Z}= \mat{A}^{-1}\mat{B}\mat{Q}^{-1}\mat{C}\mat{A}^{-1}$. Therefore,
\begin{equation}\label{eq:invextract}
\mat{A}^{-1} = \mat{X} - \mat{Y}\mat{W}^{-1}\mat{Z},
\end{equation}
Since $\mat{W}$ is assumed to be much smaller in size than $\mat{A}$, this gives rise to a rapid inverse formula with only index permutation needed. This method was mentioned earlier in a similar context by \citet{ZBY07}.

%
%
 


\end{document}